\documentclass{article}
\usepackage{xcolor}

\usepackage{amsmath,amssymb,amsthm}

\usepackage[utf8x]{inputenc}

\usepackage{hyperref}

\usepackage{array}

\usepackage{booktabs}

\usepackage{bm}

\usepackage[resetlabels]{multibib}
\newcites{app}{Appendix References}

\newtheorem{fact}{Fact}[section]
\newtheorem{lemma}[fact]{Lemma}
\newtheorem{theorem}[fact]{Theorem}
\newtheorem{definition}[fact]{Definition}

\newtheorem{proposition}[fact]{Proposition}

\newtheorem{remark}[fact]{Remark}

\usepackage{algorithm}
\usepackage{algorithmic}

\usepackage{graphicx}

\usepackage{subcaption}

\usepackage{todonotes}
\usepackage[markup=default,final]{changes}
\setdeletedmarkup{{\color{red}\sout{#1}}}

\date{}
\title{Connectome Smoothing via Low-rank Approximations}
\author{
Runze Tang\textsuperscript{1},
Michael Ketcha\textsuperscript{2},
Alexandra Badea\textsuperscript{3},  \\
Evan D.~Calabrese\textsuperscript{3},
Daniel S.~Margulies\textsuperscript{4},
Joshua T.~Vogelstein\textsuperscript{2,5}, \\
Carey E.~Priebe\textsuperscript{1},
Daniel L.~Sussman\textsuperscript{6*}}

\renewcommand{\Re}{\mathbb{R}}

\newcommand{\Ex}{\mathbb{E}}

\begin{document}

\maketitle




{\footnotesize
\noindent \textbf{1} Department of Applied Math \& Statistics, The Johns Hopkins University, Baltimore, MD
\\
\textbf{2} Department of Biomedical Engineering,  The Johns Hopkins University, Baltimore, MD
\\
\textbf{3} Department of Radiology, and Department of Biomedical Engineering, Duke University, Durham, NC
\\
\textbf{4} Max Planck Research Group for Neuroanatomy \& Connectivity, Max Planck Institute for Human Cognitive and Brain Sciences, Leipzig, Germany;
\\
\textbf{4} Child Mind Institute, New York, NY
\\
\textbf{5} Department of Mathematics \& Statistics, Boston University, Boston, MA
\\
* sussman@bu.edu
}

\begin{abstract}
In brain imaging and connectomics, the study of brain networks, estimating the mean of a population of graphs based on a sample is a core problem.
Often, this problem is especially difficult because the sample or cohort size is relatively small, sometimes even a single subject, while the number of nodes can be very large with noisy estimates of connectivity.
While the element-wise sample mean of the adjacency matrices is a common approach, this method does not exploit underlying structural properties of the graphs.
We propose using a low-rank method which incorporates dimension selection and diagonal augmentation to smooth the estimates and improve performance over the na\"ive methodology for small sample sizes.
Theoretical results for the stochastic blockmodel show that this method offers major improvements when there are many vertices.
Similarly, we demonstrate that the low-rank methods outperform the standard sample mean for a variety of independent edge distributions as well as human connectome data derived from magnetic resonance imaging, especially when sample sizes are small.
Moreover, the low-rank methods yield ``eigen-connectomes'', which correlate with the lobe-structure of the human brain and superstructures of the mouse brain, and enable estimation of latent connectome structure.
These results indicate that low-rank methods are an important part of the toolbox for researchers studying populations of graphs in general, and statistical connectomics in particular.
\end{abstract}

\section{Introduction}\label{sec:intro}

Estimating a population mean based on a sample of brain networks, as represented by their structural connectomes estimated from DTMRI, is a key challenge.
Generally, estimation of a population mean based on samples is at the core of statistics.
For a distribution of graphs, we define the mean graph as the expectation of the adjacency matrix, regardless of the data distribution. 
The sample mean, motivated by its intuitive appeal, the law of large numbers, and the central limit theorem, has its place as an important statistic for mean graph estimation.

However, the mean of a population of graphs is a high-dimensional object, consisting of $O(N^2)$ parameters for graphs with $N$ nodes (vertices).
When the number of samples $M$ is much smaller than $N^2$, or even $N$, estimating such a high-dimensional estimands using naive unbiased methods can lead to inaccurate estimates with very high variance.
Such small-$M$  large-$N$ occur frequently with new higher-resolution technologies as well as when the focus is on a smaller harder-to-sample subpopulation.
With MRI technologies improving in resolution and newer imaging modalities such as electron microscopy becoming more prominent, ever larger connectomes will be under study.
Similarly, in animal studies, studies of rare diseases, and when estimating means for small subpopulations, the number of sample graphs $M$ may be quite small.

Furthermore, using poor estimates  can lead to low power and accuracy for subsequent inference tasks.
By exploiting a bias-variance trade-off, it is often fruitful to develop estimators which have some bias but greatly reduced variance.

When these estimators are biased towards low-dimensional structures which well approximate the full dimensional population mean, major improvements in estimation can be realized~\cite{trunk1979problem}.
Furthermore, a more parsimonious representation of the mean improves interpretability and allows for novel exploratory analyses.

In statistical connectomics, \cite{ginestet2014hypothesis,Bhattacharyya2018-wj} propose a way to test if there is a difference between the distributions for two groups of networks.
While hypothesis testing is the end goal of their work, estimation is a key intermediate step which may be improved by accounting for underlying structure in the mean matrix.
Thus, improving the estimation procedures for the mean graph is not only important by itself, but also can be applied to help improve other statistical inference procedures.
Indeed, for the task of community detection, \cite{Bhattacharyya2018-wj} propose the use of low-rank mean graphs to find clusters across networks.

The entry-wise sample mean is a reasonable estimator if one is unwilling to take any additional structure into account.
However, with only a small sample size, the entry-wise sample mean does not perform very well due to high variance.
Intuitively, an estimator incorporating known structure in the true distribution of the graphs, assuming the estimator is computationally tractable, is preferable to the entry-wise sample mean.

We propose an estimator based on a low-rank structure of a family of random graphs.
Section~\ref{sec:phat} discusses details about this estimator.
These estimates can improve performance since they have much lower overall variance than naive entry-wise sample means, thereby offsetting the bias towards the low-rank structure in terms of overall error.
Additionally, \cite{Udell2016-yj} has argued that even many full-rank statistical models may be well approximated by low-rank structures.

The proposed method builds on the idea of approximating \added{a} matrix with a low-rank matrix \cite{Eckart1936-tt} but also incorporates aspects specific to networks such as bounded entries and missing diagonal entries.
Selection of the rank for the approximation is a notoriously difficult problem and our work considers two methods for automatic rank selection.

Other matrix decompositions such as non-negative matrix factorization~\cite{Berry2007-ok} and non-negative SVD~\cite{Boutsidis2008-mx} aim to decompose a matrix into non-negative factors.
This often has an intuitive appeal when some true latent factors are themselves non-negative, such as for dictionary learning tasks.
As our primary motivation is estimating the mean graph, non-negativity constraints are not needed and we are still able to use lower-dimensional factors for further exploration.
Alternatively, from the motivation of clustering networks themselves, \cite{Lee2014-vm}  has used non-negative matrix factorization of of vectorized adjacency matrices as a preliminary dictionary learning step.

The use of low-rank and spectral methods is not new to connectomics, neuroscience, and biology \cite{Banerjee2007-mk}.
\cite{Rahim2017-ci} employs low-rank shrinkage for estimation of functional connectivity. 
Connectome harmonics \cite{Atasoy2016-ip} have been investigated to understand the dynamics of neural networks \cite{Margulies2016-jj}.
The  Laplacian spectrums of connectomes have been used for comparison across species \cite{De_Lange2014-hd}.
Our work specifies when and how these types of methods can be used to improve performance of mean graph estimation.

Importantly, low-rank methods also produce parsimonious and interpretable representations of the data as represented by the random dot product graph model.
The proposed algorithm can be viewed as yielding estimates of low-rank latent structure in the population, regardless of whether the true mean graph is low-rank or not.
As we demonstrate, the interpretations of these latent parameter estimates
illustrate hypotheses which relate the structure of the connectome to well established anatomical structures (see Fig.~\ref{fig:eigenvector_brain} and Section~\ref{section:lobe_structure}) and suggest a basis for mapping structural hierarchies in brain organization.

\section{Statistical Connectome Models}
\label{section:model}

In connectomics, brain imaging data for each subject can be processed  to output a graph, where each vertex represents a well-defined anatomical region present in each subject.
For structural brain imaging, such as diffusion tensor MRI, an edge may represent the presence of anatomical connections between the two regions as estimated using tractography algorithms~\cite{gray2012magnetic}.
For functional brain imaging, such as fMRI, an edge between two regions may represent the presence of correlated brain activity between the two regions.

This work considers the scenario of observing $M$ graphs, represented as adjacency matrices, $A^{(1)},A^{(2)},\dotsc,A^{(M)}$, each having $N$ vertices with $A^{(m)}\in\{0,1\}^{N\times N}$ for each $m$.
We assume there is a known correspondence for vertices in different graphs, so that vertex $i$ in graph $m$ corresponds to vertex $i$ in graph $m'$ for any  $i$, $m$, $m'$.
The graphs we consider are undirected and unweighted with no self-loops, so each $A^{(m)}$ is a binary symmetric matrix with zeros along the diagonal.


We assume that the graphs are sampled independently and identically from some distribution andthe mean graph is the expectation of each adjacency matrix.
\begin{definition}[Mean Graph]
Suppose $A^{(1)},\dotsc,A^{(M)}\stackrel{iid}{\sim} \mathcal{G}$ for some random graph distribution $\mathcal{G}$, with $A^{(m)}\in\{0,1\}^{N\times N}$ for each $m$.
The {\em mean graph} is $\Ex[A^{(m)}]$, which we will denote by $P$.
\end{definition}
Note, this definition requires no assumption on the distribution $\mathcal{G}$.

\subsection{Independent Edge Model}
The most general model we consider is the independent edge model (IEM) with parameter given by a mean graph $P \in [0,1]^{N\times N}$ \cite{bollobas2007phase}.
An edge exists between vertex $i$ and vertex $j$ with probability $P_{ij}$, and each edge is present independently of all other edges.
Note that the IEM is a generalization of the Erd\H{o}s-R{\'e}nyi random graphs, where each edge is present with probability $p$ independently of all other edges \cite{Gilbert1959-ba,Erdos1959-ln}.

The assumption of independent edges, while not necessary to define the mean graph, is important for our theory.
In practice the independent edge assumption may not hold, as brain graphs will have complex spatial and functional dependencies.
On the other hand, if $\Ex[A^{(m)}]$ has low-rank structure then our methods below can still improve performance.
Regardless of the structure $P$, the low-rank method below will provide an estimate of the low-rank structure of $P$.

\subsection{Random Dot Product Graph}

The random dot product graph model (RDPG) \cite{young2007random, nickel2007random}, assigns latent positions to each vertex in a graph where the probability of an edge being present between two nodes is the dot product of their latent vectors \cite{hoff2002latent}.

%
Formally, let $\mathcal{X} \subset \Re^d$ be a set such that $x, y \in \mathcal{X}$ implies $x^{\top} y =\sum_{i = 1}^d x_i y_i \in [0, 1]$.
Let $X_1,\dotsc,X_N\in \mathcal{X}$ be iid column vectors representing the $N$ latent positions and let $X = [X_1|\cdots|X_N]^{\top} \in \Re^{N \times d}$.
Note, the entries of $X$ are not required to be non-negative and for these purposes this introduce no interpretability issues.
We assume that $X$ is the same for all observed graphs.

A  random adjacency matrix $A$ is said to be an RDPG if for each adjacency matrix $a\in\{0,1\}^{n\times n}$,
\[
    \Pr[A = a|X] = \prod_{i<j} (X_i^{\top} X_j)^{a_{ij}} ( 1 - X_i^{\top} X_j)^{1 - a_{ij}}.
\]

Conditioned on the latent positions, the RDPG is an independent edge model with probability matrix given by the outer product of the latent position matrix with itself, $P = X X^{\top}$.
This imposes that $P$ is positive-semidefinite and $\mathrm{rank}(P)=\mathrm{rank}(X)\leq d$.

\subsection{Stochastic Blockmodel as an RDPG}
\label{section:sbm_rdpg}

One of the most common structures for graphs is that vertices tend to cluster into communities.
Vertices of the same community behave similarly, connecting to similar sets of nodes.

This structural property is captured by the stochastic blockmodel (SBM) \cite{holland1983stochastic}, where each vertex is assigned to a block and the probability that an edge exists between two vertices depends only on their respective block memberships.


The SBM is parameterized by the number of blocks $K$ (generally much less than the number of vertices $N$), the block probability matrix $B \in [0,1]^{K \times K}$, and the vector of iid block memberships
$\tau\in\{1,\dotsc,K\}^N$. For each $i \in [N]$, $\tau_i = k$ means vertex $i$ is a member of block $k$.
Denote as $\rho_k$ the probability $\Pr[\tau_i = k]$ for each $k$.
We assume that the $\tau$ vector is the same for all observed graphs.

Conditioned on $\tau$, $A_{ij} \stackrel{ind}{\sim} \mathrm{Bern}(B_{\tau_i,\tau_j})$.
To ensure that the SBM can be considered as an RDPG, we always impose that the $B$ matrix for the SBM is positive semidefinite (see Appendix~\ref{app:outline_proof}).



To better describe complex network structures, many
generalizations of the SBM incorporate the variation of vertices within blocks. 
For example, \cite{airoldi2008mixed} proposed mixed membership stochastic blockmodels, 
and \cite{karrer2011stochastic} proposed degree-corrected SBM.
These generalizations aim to capture variations among vertices while maintaining parts of the original community structure.
The RDPG is useful in this regard since any SBM with degree-correction and mixed-membership can be represented as an RDPG and visa versa \cite{Lyzinski2014-az,Rubin-Delanchy2017-av}.

\section{Estimators}
\label{sec:estimator}

\subsection[Element-wise Sample Mean]{Element-wise Sample Mean $\bm{\bar{A}}$}
\label{sec:abar}

The most natural estimator to consider is the element-wise sample mean.
This estimator, defined as $\bar{A}=\frac{1}{M}\sum_{m=1}^M A^{(m)}$, is the  maximum likelihood estimator (MLE) for the mean graph $P$ if the graphs are sampled from an IEM distribution.
It is unbiased so $\Ex[\bar{A}]=P$ with entry-wise variance $\mathrm{Var}(\bar{A}_{ij}) = P_{ij} (1-P_{ij})/M$. 
Moreover, for the independent edge model, $\bar{A}$ is the uniformly minimum-variance unbiased estimator, so it has the smallest variance among all unbiased estimators.
Similarly, it enjoys the many asymptotic properties of the MLE as $M\to \infty$ for fixed $N$.
However, if graphs with a large number of vertices are of interest, $\bar{A}$ is not consistent for $P$ as the number of vertices $N$ becomes large for fixed $M$, while our estimator $\hat{P}$ from Section~\ref{sec:phat} is consistent for low-rank $P$.

Additionally, $\bar{A}$ does not exploit any low-dimensional structure.
If the graphs are distributed according to an RDPG or SBM, then $\bar{A}$ is no longer the maximum likelihood estimator since it is not guaranteed to satisfy the properties of the mean graph for that model.
In the RDPG case, \deleted{since} $\mathrm{rank}(\bar{A})$ will with high probability not be equal to $d$, and for the SBM case, $\bar{A}$ will not have the structure of an SBM mean graph, with $K$ distinct rows and columns.
Hence, in either case $\bar{A}$ will be outside the parameter space and idadmissible.

The performance can be especially poor when the sample size $M$ is small, such as when $M\ll N$.
For example, when $M=1$, $\bar{A}$ is simply the binary adjacency matrix $A^{(1)}$, which is an inaccurate estimate for an arbitrary $P$ compared to estimates which exploit underlying structure, such as the low-rank structure of the RDPG model.

\subsection[Low-Rank Estimator]{Low-Rank Estimator $\bm{\hat{P}}$}
\label{sec:phat}

Motivated by the low-rank structure of the RDPG mean matrix, we propose the estimator $\hat{P}$ based on the spectral decomposition of $\bar{A}$, yielding a low rank approximation of $\bar{A}$.
Low-rank methods for matrices can be viewed as being closely related to ideas such as principal components analysis.
Viewed from a regularization perspective, the low-rank approximation finds an estimate with many zero-eigenvalues which can be viewed as analogous to an $L_0$ constraint on the eigenvalues of the solutions.
These estimates are also related to methods employing nuclear norm penalization~\cite{chatterjee2015matrix}.
%

This estimator is similar to the estimator proposed by~\cite{chatterjee2015matrix} but incorporates additional adjustments which serve to improve the performance for the specific task of estimating the mean graph.
Additionally, we consider an alternative dimension selection technique.
Details of the dimension selection procedures and the diagonal augmentation are in Appendices~\ref{app:dim_select} and \ref{app:diag_aug}, respectively.
To summarize, the overall strategy to compute $\hat{P}$ is described in Algorithm~\ref{algo:basic}.
A key component of this algorithm is the low-rank approximation from Algorithm~\ref{algo:lowrank}.

\begin{algorithm}[H]
\caption{Rank-$d$ approximation of a matrix.}
\label{algo:lowrank}
\begin{algorithmic}[1]
\REQUIRE Symmetric matrix $A\in \Re^{N\times N}$, dimension $d\leq N$.
\ENSURE $\mathrm{lowrank}_d(A)\in \Re^{N\times N}$
\STATE Compute the algebraically largest $d$ eigenvalues of $A$, $s_1\geq s_2\geq \dotsc\geq s_d$ and corresponding orthonormal eigenvectors $u_1,u_2,\dotsc,u_d\in \Re^N$;
\STATE $\hat{S} \leftarrow \mathrm{diag}(s_1,\dotsc,s_d)$ and  $\hat{U} \leftarrow [u_1,\dotsc,u_d]$;
\STATE Return $\hat{U}\hat{S}\hat{U}^{\top}$;
\end{algorithmic}
\end{algorithm}

The first step is to calculate the sample mean $\bar{A}$.
In Step 2 to Step 4, the algorithm augments the diagonal of $\bar{A}$ based on \cite{marchette2011vertex}, selects the dimension $\hat{d}$ to embed (see Appendix~\ref{app:dim_select}), and computes the low-rank approximation $\tilde{P}^{(0)}$ based on the embedding. 
Then in Step 5 and Step 6, the algorithm augments the diagonal again based on~\cite{scheinerman2010modeling} (see Appendix~\ref{app:diag_aug}) which yields an improved low-rank estimate $\tilde{P}^{(1)}$.
Finally, Step 7 thresholds the matrix entries to ensure all elements are between 0 and 1.

\begin{algorithm}
\caption{Algorithm to compute $\hat{P}$}
\label{algo:basic}
\begin{algorithmic}[1]
\REQUIRE Adjacency matrices $A^{(1)}, A^{(2)}, \cdots, A^{(M)}$, with each $A^{(m)} \in \{0,1\}^{N \times N}$
\ENSURE Estimate $\hat{P}\in[0,1]^{N\times N}$
\STATE $\bar{A} \leftarrow \left(\sum\limits_{m = 1}^M A^{(m)}\right)/M$;
\STATE $D^{(0)} \leftarrow \mathrm{diag}(\bar{A} \bm{1})/(N-1)$;
\STATE $\hat{d} \leftarrow \mathrm{dimselect}(\bar{A} + D^{(0)})$; (see Appendix~\ref{app:dim_select})
\STATE $\tilde{P}^{(0)} \leftarrow \mathrm{lowrank}_{\hat{d}}(\bar{A} + D^{(0)})$; (see Algorithm~\ref{algo:lowrank})
\STATE $D^{(1)} \leftarrow \mathrm{diag}(\tilde{P}^{(0)})$;
\STATE $\tilde{P}^{(1)} \leftarrow \mathrm{lowrank}_{\hat{d}}(\bar{A} + D^{(1)})$; (see Algorithm~\ref{algo:lowrank})
\STATE $\hat{P} \leftarrow \min(\max(\tilde{P}^{(1)}, 0), 1)$.
\end{algorithmic}
\end{algorithm}



For a given dimension $d$ we consider the estimator $\mathrm{lowrank}_d(\bar{A})$ defined as the best rank-$d$ positive-semidefinite approximation of $\bar{A}$.
Let $\hat{S}$ be a diagonal matrix with non-increasing entries along the diagonal corresponding to the largest $d$ eigenvalues of $\bar{A}$ and let $\hat{U}$ have columns given by the corresponding eigenvectors. Similarly, let $\tilde{S}$ be the diagonal matrix with non-increasing entries along the diagonal corresponding to the remaining $N - d$ eigenvalues of $\bar{A}$ and let $\tilde{U}$ have columns given by the corresponding eigenvectors.
Since the graphs are symmetric, the eigen-decomposition can be computed as $\bar{A}$ as $\hat{U} \hat{S} \hat{U}^{\top} + \tilde{U}\tilde{S}\tilde{U}^{\top}=[\hat{U}|\tilde{U}] (\hat{S}\oplus \tilde{S}) [\hat{U}|\tilde{U}]^T$.
The $d$-dimensional adjacency spectral embedding (ASE) of $\bar{A}$ is given by $\hat{X}=\hat{U} \hat{S}^{1/2}\in \Re^{N \times d}$.
For an RDPG, the rows of $\hat{X}$ are estimates of the latent vectors for each vertex \cite{sussman2014consistent}.
Using the adjacency spectral embedding, the low-rank approximation \replaced{of}{is} $\bar{A}$ \replaced{is}{to be} $\hat{X} \hat{X}^{\top}=\hat{U}\hat{S}\hat{U}^{\top}$.
Algorithm~\ref{algo:lowrank}  gives the steps to compute this low-rank approximation for a general symmetric matrix $A$.


To compute the estimator $\hat{P}$, the rank $d$  must be specified; there are various ways of dealing with dimension selection.
In this paper, we explore an elbow selection method proposed in \cite{zhu2006automatic} and the universal singular value thresholding (USVT) method \cite{chatterjee2015matrix}.
Appendix~\ref{app:dim_select} discusses details of these methods.

Moreover, when the adjacency matrices are hollow, with zeros along the diagonal, there is a missing data problem that leads to inaccuracies if $\hat{P}$ is computed based only on $\bar{A}$.
To compensate for this issue, we use an iterative method developed in \cite{scheinerman2010modeling}.
Appendix~\ref{app:diag_aug} discusses details of the iterative method.



\section{Theoretical Results}
\label{section:theoretical_result}

To estimate the mean of a collection of graphs, we compare the two estimators from Section~\ref{sec:estimator}: the entry-wise sample mean $\bar{A}$ and the low-rank $\hat{P}$ motivated by the RDPG.
The mean squared errors (MSE) for our estimators are $\mathrm{MSE}(\hat{P}_{ij})=\Ex[\hat{P}_{ij}-P]^2$ and $\mathrm{MSE}(\bar{A})=\Ex[\bar{A}_{ij}-P]^2$.
The relative efficiency for two estimators is the ratio of their MSE, $\mathrm{RE}(\bar{A}_{ij},\hat{P}_{ij}) = \frac{\mathrm{MSE}(\hat{P}_{ij})}{\mathrm{MSE}(\bar{A}_{ij})}$, with values above 1 indicating $\bar{A}$ should be preferred while values below 1 indicate $\hat{P}$ should be preferred.
Relative efficiency is a useful metric for comparing estimators because it will frequently be invariant to the scale of the noise in the problem and hence is comparable across different settings.

\subsection{SBM}

In this section, entry-wise relative efficiency is computed to analyze the performance of these two estimators under the SBM.
The asymptotic relative efficiency is defined as $\lim_{N\to \infty}\mathrm{RE}$.
We also define the scaled relative efficiency, $N\cdot \mathrm{RE}(\bar{A}_{ij},\hat{P}_{ij})$ which normalizes the relative efficiency so that the asymptotic scaled relative efficiency is non-zero and finite.
Somewhat surprisingly, the results indicate that the asymptotic relative efficiency will not depend on the sample size $M$.


For this asymptotic framework, 
the proportion of vertices in block $k$, $|\{i:\tau_i=k\}|/N$, will converge to $\rho_k$ as $N\to\infty$ by the law of large numbers.

Denote the block probability matrix as $B = \nu \nu^{\top} \in [0, 1]^{K \times K}$.
By definition, the mean of the collection of graphs generated from this SBM is $P \in [0, 1]^{N \times N}$, where $P_{ij} = B_{\tau_i, \tau_j}$. After observing $M$ graphs on $N$ vertices $A^{(1)}, \cdots, A^{(M)}$ sampled independently from the SBM conditioned on $\tau$, the two estimators can be calculated, $\bar{A}$ and $\hat{P}$.

\begin{lemma}
\label{lm:VarPhat}
For the above setting, for any $i \ne j$, if $\mathrm{rank}(B)=K=d$; for large enough $N$,
\[
    \Ex[(\mathrm{lowrank}_d(\replaced{\bar{A}}{P})_{ij} - P_{ij})^2] \approx
    \frac{1/\rho_{\tau_i} + 1/\rho_{\tau_j}}{M N} P_{ij}(1-P_{ij}),
\]
and
\[
    \lim_{N \to \infty} N \cdot \mathrm{Var}(\mathrm{lowrank}_d(\replaced{\bar{A}}{P})_{ij}) =
    \frac{1/\rho_{\tau_i} + 1/\rho_{\tau_j}}{M} P_{ij} (1 - P_{ij}).
\]
\end{lemma}

\begin{theorem}
\label{thm:ARE}
In the same setting as in Lemma~\ref{lm:VarPhat}, for any $i \ne j$, if $\mathrm{rank}(B)=K=d$, then for large enough $N$:
\begin{equation}
	    \mathrm{RE}(\bar{A}_{ij}, \mathrm{lowrank}_d(\replaced{\bar{A}}{P})_{ij}) \approx
    \frac{1/\rho_{\tau_i} + 1/\rho_{\tau_j}}{N},
\label{eq:approx_re}
\end{equation}
and the asymptotic relative efficiency (ARE) is
\[
    \mathrm{ARE}(\bar{A}_{ij}, \mathrm{lowrank}_d(\replaced{\bar{A}}{P})_{ij}) = \lim_{N \to \infty} \mathrm{RE}(\bar{A}_{ij}, \mathrm{lowrank}_d(\replaced{\bar{A}}{P})_{ij}) = 0.
    \label{eq:sbm_are}
\]
\end{theorem}

Note these theorems are stated for $\mathrm{lowrank}_d(\replaced{\bar{A}}{P})$ rather than $\hat{P}$.
This allows us to employ theoretical developments which apply specifically to the rank-$d$ approximation of $\bar{A}$ \cite{athreya2016limit}.
While we assume that $\mathrm{rank}(P)$ is known for this theory, asymptotically accurate methods to estimate $d$ in the SBM and RDPG settings have been established \cite{Fishkind2012}.
In the large $N$ setting, the other elements of Algorithm~\ref{algo:basic}, notably the diagonal augmentations and thresholding, result in negligible deviations of from $\mathrm{lowrank}_d(\replaced{\bar{A}}{P})$ and hence do not impact the asymptotic theory.
Nonetheless, these components do improve finite sample performance, see Appendix~\ref{app:compare_param}.

Note that $\rho_{\tau_i}$ represents the probability that a vertex is assigned to the same block as vertex $i$, i.e.\ $\tau_i$-th block.
The proofs of these results are provided in Appendix~\ref{app:outline_proof}. 


This theorem indicates that under the SBM, $\hat{P}$ is a much better estimate of the mean of the collection of graphs $P$ than $\bar{A}$, especially when $N$ is large.
Note that a relative efficiency less than 1 indicates that $\hat{P}$ should be preferred over $\bar{A}$, so under the above assumptions, as $N\to\infty$, $\hat{P}$ performs far better than $\bar{A}$.
Note that even though the RE could be greater than $1$ for some $N$, eventually the RE will go to 0 as $N$ increases.
The result shows that the relative efficiency is of order $O(N^{-1})$ and $N \cdot \mathrm{RE}(\bar{A}_{ij}, \hat{P}_{ij})$ (denoted as scaled RE) converges to $1/\rho_{\tau_i}+1/\rho_{\tau_j}$ when $N\to\infty$.

An important aspect of Theorem~\ref{thm:ARE} is that the ARE does not depend on the number of graphs $M$, so the larger the graphs are, the better $\hat{P}$ is relative to $\bar{A}$, regardless of $M$.
For smaller values $N$, the impact of $M$ may be more substantial.
In the SBM and RDPG cases, 
$\hat{P}$ should still offer substantial improvements over $\bar{A}$.
The theory in this case is more difficult to analyze as the impacts of low-rank approximations in the small $N$ setting are less well understood (see Section~\ref{sec:low_rank}).

The asymptotic results are for a number of vertices going to infinity with a fixed number of graphs, a setting which will be very useful in future connectomics analysis, as the collection of larger and larger brain networks grow from small sample sizes.
For example, \cite{Calabrese2015} recently reported a high resolution magnetic resonance microscopy based estimate of the mouse brain using a single mouse.

The approximate formula Eq.~\ref{eq:approx_re} indicates that the sizes of the blocks can greatly impact the relative efficiency.
As an example, consider a 2-block SBM \added{with large but fixed number of nodes $N$}.
If each of the blocks contain half the vertices, then for each pair of vertices, the relative efficiency is approximately $4/N$.
If the first block gets larger, with $\rho_1\to 1$, then the RE for estimating $P_{ij}$ with $\tau_i=\tau_j=1$ will tend to its minimum of $2/N$.
On the other hand as $\rho_1\to 1$, if $\tau_i=1$ and $\tau_j=2$, then, since $\rho_2=1-\rho_1$, the relative efficiency for estimating such an edge pair will be approximately $1$ and the same will hold if $\tau_i=\tau_j=2$.
Note that the maximum value for the relative efficiency of two vertices from different blocks in a two-block model is achieved when $\rho_1=1/N$ and $\rho_2=(N-1)/N$ in which case the relative efficiency is $N/(N-1) \approx 1$.
(Values of $\rho_s$ below $1/N$ correspond to graphs where no vertices are typically in that block, so the effective minimum that can be considered for $\rho_s$ is $1/N$.)
Note, $N \cdot \mathrm{RE}(\bar{A}_{ij}, \hat{P}_{ij})$ achieves its minimum for $i$ and $j$ from different blocks when $\rho_k = 1/K$ for all $k$.

Finite sample simulations illustrating these results are in Appendix~\ref{app:sbm_sim}.

\subsection{RDPG}\label{sec:rdpg_theory}

If instead of assuming that the graphs follow an SBM distribution, we assume that the graphs are distributed according to an RDPG distribution, similar gains in relative efficiency can be realized.
While there is no compact analytical formula for the relative efficiency of $\hat{P}$ versus $\bar{A}$ in the general RDPG case, using the same ideas as in Theorem~\ref{thm:ARE}, we can show that $\mathrm{RE}(\bar{A}_{ij},\hat{P}_{ij}) = O(1/N)$.

\begin{proposition}
Suppose that $A^{(1)},A^{(2)},\dotsc,A^{(M)}$ are iid from an RDPG distribution with common latent positions $X_1,\dotsc,X_n$, which are drawn iid from a fixed distribution.
As the number of vertices $N\to\infty$, it holds for any $i\neq j$ that $\mathrm{RE}(\bar{A}_{ij},\replaced{\mathrm{lowrank}_d(\bar{A})}{\hat{P}}_{ij}) = O(1/N)$,
where again the asymptotic relative efficiency in $N$ does not depend on $M$.
\end{proposition}
The proof of this proposition closely follows the proofs of Lemma~\ref{lm:VarPhat} and Theorem~\ref{thm:ARE}.

\subsection{Generalizations}\label{sec:low_rank}

When low-rank assumptions hold for the mean graph $P$, our theory and subsequent simulations show that $\hat{P}$ will significantly outperform $\bar{A}$.
However, if $P$ is full-rank or nearly full-rank, then we do not have such guarantees.

If the graphs are distributed according to an SBM or an RDPG, the relative efficiency is approximately invariant to the number of graphs $M$ when $N$ is large.
If on the other hand, the graphs are generated according to a full-rank independent edge model, then the relative efficiency can change more dramatically as $M$ changes.
For larger $M$, more of the eigenvectors of $\bar{A}$ will  concentrate around the eigenvectors of the mean graph.
This leads to the fact that the optimal embedding dimension for estimating the mean will increase, making $\bar{A}$ and the optimal low-rank approximation more similar.
As a result, $\mathrm{RE}(\bar{A},\hat{P})$ will increase as $M$ increases for full-rank models, with $\mathrm{RE}(\bar{A},\hat{P})$ possibly $\geq 1$ since it is not guaranteed that $\hat{P}$ will choose the optimal dimension.

As $M\to \infty$ with $N$ fixed, the optimal embedding dimension will itself tend to $N$.
If the dimension selection method also tends to $N$, such as for the USVT method, then for $M$ large, $\hat{P}$ and $\bar{A}$ will coincide.
Note, this relies on also including negative eigenvalues in the estimate.
In the general case, for fixed $M$ and $N$, whether $\hat{P}$ or $\bar{A}$ has better performance is a difficult theoretical question beyond the scope of this manuscript.

Another challenge in some network contexts is sparsity of the graph as $N$ grows.
Many of the theoretical underpinnings of our proofs have been extended to this setting \cite{Tang2016-rs}, which enable the extension of our results to the sparse setting for average degrees as small as $\theta(\log^4 N)$.

\section{Human Connectomes}\label{sec:connectome}

In practice, observed graphs do not follow the independent edge model, let alone an RDPG or SBM, but the mean of a population of graphs is still of interest.
To demonstrate that the estimator $\hat{P}$ is  useful in such cases, its performance on structural connectomic data is tested.
The graphs are based on diffusion tensor MR images of the SWU4 dataset collected and available at the Consortium for Reliability and Reproducibility \cite{zuo2014open} (see Appendix~\ref{app:data_human} for dataset details).
The dataset contains 454  brain scans, each of which was processed to yield an undirected, unweighted graph with no self-loops, using the pipeline described in \cite{kiar2017science, kiar2016ndmg}.
The vertices of the graphs represent  regions in the brain defined according to an atlas.
Here, three atlases were used: the JHU atlas with 48 vertices \cite{oishi2010mri}, the Desikan atlas with 70 vertices \cite{desikan2006automated}, and the  CPAC200 atlas with 200 vertices \cite{sikka2014towards}.
An edge exists between two vertices whenever there is at least one white-matter tract connecting the corresponding two regions of the brain.

\begin{figure}[!tbp]
\centering
\includegraphics[width=\linewidth]{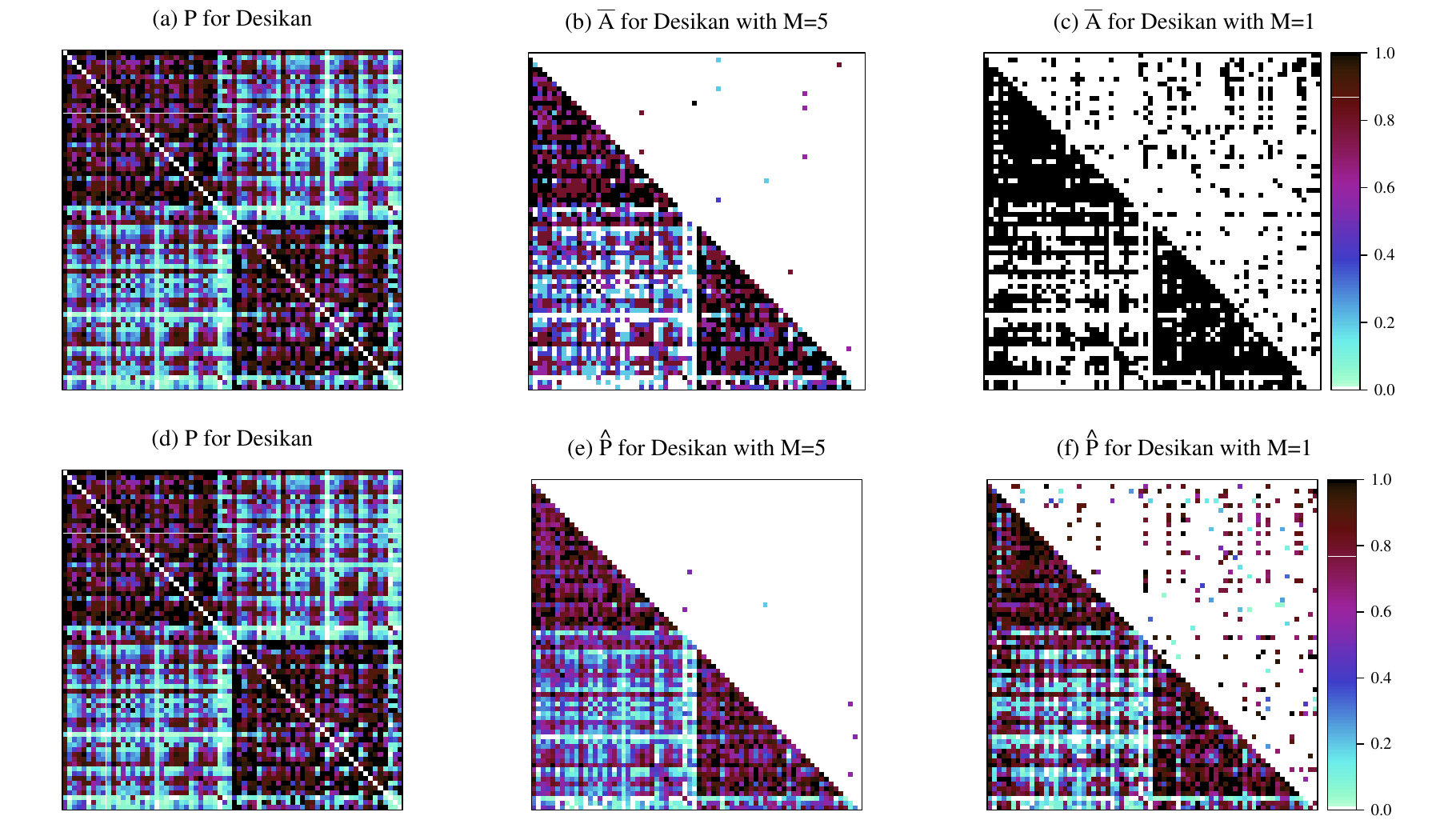}
\caption{
The heat map of the population mean for the $454$ human connectome graphs from the SWU4 dataset and the Desikan atlas are depicted in Panel (a) (also in Panel (d)).
Darker pixels indicate a higher proportion of graphs having an edge between the given vertices.
The two heat maps in the second column indicate the sample mean for 5 sampled graphs (Panel (b)), and the low-rank estimate $\hat{P}$ for the same 5 graphs with rank $d=11$ (Panel (e)), computed using Algorithm~\ref{algo:basic}.
The mean squared errors (MSE) for this sample are  $MSE(\hat{P})=0.015$ as compared to $MSE(\bar{A})=0.016$, a 3\% relative improvement.
To highlight the improvements, the upper triangular areas of the heat maps for $\bar{A}$ and $\hat{P}$ show the 18 edges  for $\bar{A}$ and 6 edges for $\hat{P}$ which have absolute estimation error larger than $0.4$.
In the third column, two heat maps using a sample size $M = 1$  ($\hat{P}$ is calculated with a dimension $d = 12$) show a smoothing effect in the heat map of $\hat{P}$ (MSE$=0.049$), which leads to a 53\% relative improvement compared to $\bar{A}$(MSE$=0.104$). Similarly, the same absolute estimation error threshold of $0.4$ highlights 504 edges for $\bar{A}$ and 234 edges for $\hat{P}$.}
\label{fig:Matrix_desikan_m5}
\end{figure}

Our goal is to estimate the mean graph of the population $P$, defined as the entry-wise mean of all the 454 graphs. Fig.~\ref{fig:Matrix_desikan_m5} (a) shows a heat map of the population mean graph $P$.
Darker pixels indicate a higher proportion of graphs having an edge between the given vertices.
Fig.~\ref{fig:Matrix_desikan_m5}(b) depicts the entry-wise sample mean $\bar{A}$ when the sample size is $M=5$ in the SWU4 dataset example.
While $\bar{A}$ is a reasonable estimate of $P$, there are some vertex-pairs with very inaccurate estimates.
The upper triangular area of the heat map for $\bar{A}$ depicts the 18 vertex-pairs which have an absolute estimation error larger than $0.4$.
When the sample size is small, the performance of $\bar{A}$ degrades due to its high variance.
Such phenomena are most obvious when the sample size decreases from $M = 5$ to $M = 1$.
Fig.~\ref{fig:Matrix_desikan_m5}(c) shows the heat map of $\bar{A}$ based on sample size $M = 1$.
Since there is only one observed graph, $\bar{A}$ is binary and thus very bumpy.
Similarly, when the same absolute estimation error threshold is $0.4$, 504 (out of 2415) edges in the upper triangular area are highlighted.

When we use the same random sample size of $M=5$ as in Fig.~\ref{fig:Matrix_desikan_m5}, the plot of $\hat{P}$ in Panel (e) shows a finer gradient of values which results in a 3\% relative improvement in estimation of the true probability matrix, $P$.
$\bar{A}$ has mean squared error of $0.016$ and $\hat{P}$ has mean squared error of $0.015$.
The upper triangular area of the heat map for $\hat{P}$ depicts the 6 edges which have absolute estimation error larger than $0.4$, whereas 18 edges are highlighted for $\bar{A}$ based on the same threshold.

The smoothing effect is even more obvious when $M = 1$, as  in Fig.~\ref{fig:Matrix_desikan_m5}(f). $\hat{P}$ smooths the estimate, especially for edges across the two hemispheres, in the lower left and corresponding upper right block (which is not shown in the heat map).
Based on the calculations, $\hat{P}$, with mean squared error $0.049$, outperforms $\bar{A}$, with mean squared error $0.104$, a 53\% relative improvement in estimation.
Similarly, the same absolute estimation error threshold of $0.4$ highlights 234 edges for $\hat{P}$, less than $50\%$ as many as $\bar{A}$.




A cross validation on the 454 graphs serves to evaluate the performance of the two estimators.
Specifically, for each atlas, each Monte Carlo replicate consists of sampling $M$ graphs out of the 454, and computing the low-rank estimator $\hat{P}$ and the sample mean $\bar{A}$ on the $M$  graphs.
These estimates are compared to the sample mean $P$ for all $454$ adjacency matrices.


To evaluate performance, the average of the ratios of the mean squared error across all vertex pairs is computed.
1000 cross-validation simulations on each of the three atlases are run for sample sizes of $M=1, 5, 10$. For $M=1$, only the 454 distinct possibilities are considered.
To determine the rank for $\hat{P}$, we employed Zhu and Ghodsi's method \cite{zhu2006automatic} and USVT \cite{chatterjee2015matrix} (see Appendix~\ref{app:dim_select}).

\begin{figure}
\begin{center}
  \includegraphics[width=1\linewidth]{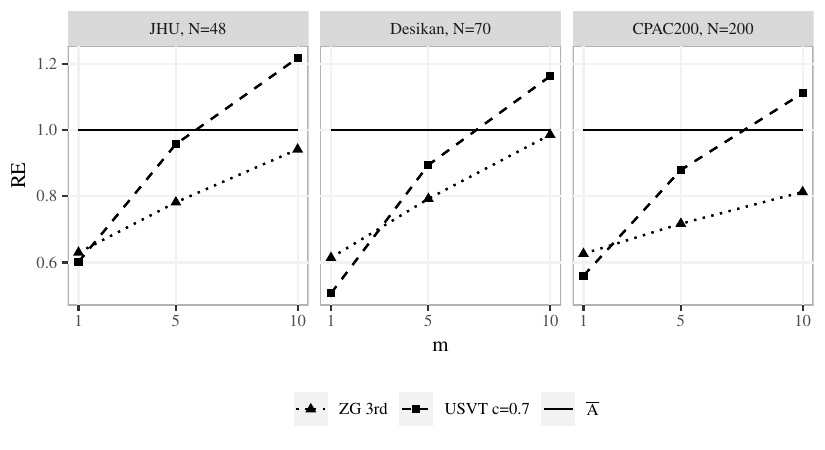}
\end{center}
\caption{
$\bar{A}$ and $\hat{P}$ were computed on samples of graphs from each atlas: JHU, Desikan, and CPAC200, with different sample sizes $M$ and different dimension selection procedures, ZG and USVT.
For each of the two methods for computing $\hat{P}$, relative efficiencies were estimated with respect to the sample mean $\bar{A}$.
Confidence intervals all had lengths less than $0.015$, and hence were omitted for clarity.
When $M=1$ or 5, $\hat{P}$ always provides substantial improvements (RE $<$ \replaced{.7}{0.85}) compared to $\bar{A}$.
For $M=10$, $\hat{P}$ \added{using USVT} has worse performance than $\bar{A}$ \deleted{for the two smaller atlases} but \added{using ZG still} improves upon $\bar{A}$\deleted{ for the CPAC200 atlas}.
}
\label{fig:corr_re}
\end{figure}

Fig.~\ref{fig:corr_re} shows the estimated relative efficiencies between $\bar{A}$ and $\hat{P}$.
For each atlas and each sample size, both dimension selection methods have similar overall performance.
Confidence intervals for the estimated relative efficiencies, calculated by assuming a normal distribution, all have lengths less than $0.015$.
\deleted{Except for the CPAC200 atlas with $M=10$ using USVT, a}\added{A}ll relative efficiencies are significantly different from 1.

The largest improvements using $\hat{P}$ occur when $M$ is small and $N$ is large, where the RE are smaller than 1.
On the other hand, once $M=10$, $\bar{A}$ tends to do nearly as well or better than $\hat{P}$\added{, except for the larger atlas using ZG}.


For the sample with size $M=5$ from Fig.~\ref{fig:Matrix_desikan_m5}, Fig.~\ref{fig:Diff_desikan_m5} shows the values for the absolute estimation error $|\bar{A} - P|$ and $|\hat{P}-P|$, as well as $|\bar{A} - \hat{P}|$.
The lower triangular sections show the  absolute differences while the upper triangular matrix highlights vertex pairs with absolute differences larger than 0.4.
There are 18 edges for $\bar{A}$ and only 6 edges for $\hat{P}$ being highlighted in the figure.
Note that approximately $13\%$ of all pairs of vertices are adjacent in all $454$ graphs and hence $\bar{A}$ will always have zero error for those pairs of vertices.
Nonetheless, $\hat{P}$ typically outperforms $\bar{A}$.

\begin{figure}
\begin{center}
  \includegraphics[height=.325\linewidth]{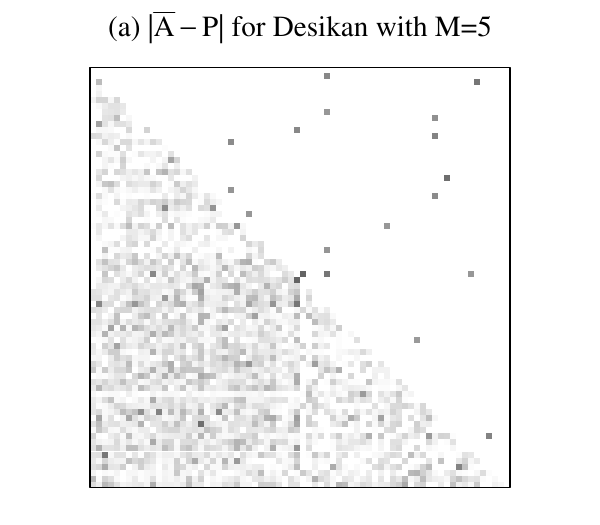}\hspace{-35pt}
  \includegraphics[height=.325\linewidth]{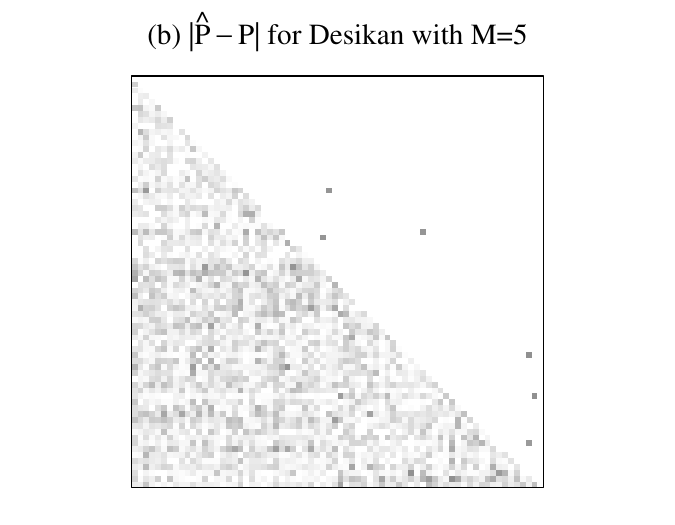}\hspace{-28pt}
  \includegraphics[height=.325\linewidth]{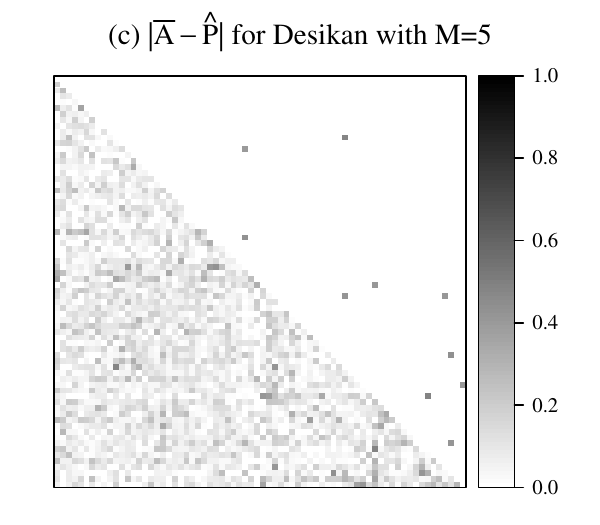}
\end{center}
\caption{
For the SWU4 data with the Desikan atlas, using a sample of $M=5$ graphs, $\bar{A}$ and $\hat{P}$ are computed (Fig.~\ref{fig:Matrix_desikan_m5} (b) and (e)).
Here, the lower triangular matrices show the absolute estimation error $|\bar{A} - P|$, $|\hat{P} - P|$ and $|\bar{A} - \hat{P}|$.
The embedding dimension for $\hat{P}$ is $d=11$ selected by the ZG method. 
The upper triangular matrix highlights the edges with absolute differences larger than $0.4$, with 18 edges from $\bar{A}$ and only 6 edges from $\hat{P}$ being highlighted.
Overall, $\hat{P}$ provides improved performance over $\bar{A}$ for this sample size.
}
\label{fig:Diff_desikan_m5}
\end{figure}


\subsection{Challenges of the SWU4 Dataset}
\label{sec:challenge}
While  $\hat{P}$ performs well when the sample size $M$ is small and the number of vertices $N$ is large, the SWU4 dataset itself does not strictly adhere to the low-rank assumptions of our theory.
Whether the dataset has strong or weak low-rank structure was investigated. In Fig.~\ref{fig:screeplot}, the relative error $\|\mathrm{lowrank}_d(P)-P\|_F^2/\|P\|_F^2$ of using a rank-$d$ approximation of $P$ (see Algorithm~\ref{algo:lowrank}) is plotted as solid curves.
The rate at which this curve tends to zero provides an indication of the relative increase in error when using $\mathrm{lowrank}_d(\replaced{\bar{A}}{P})$ as compared to $\bar{A}$, when $M$ is large.
For all three atlases, substantial errors remain for any low-rank ($<n/2$) approximation.
This can be compared to the dashed lines which show how these error increases would behave if $P$ was truly low-rank where the ranks are selected by Zhu and Ghodsi's method, 13 for JHU, 8 for Desikan, and 37 for CPAC200.

\begin{figure}[!thbp]
\centering
\includegraphics[width=\linewidth]{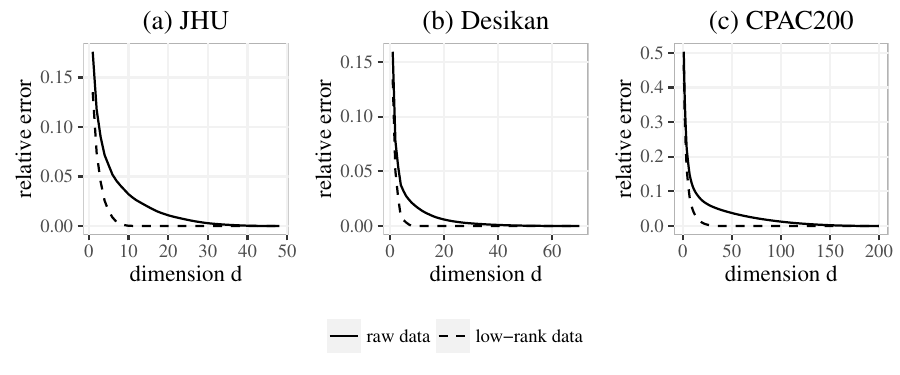}
\caption{
The solid curves show the relative error $\|\mathrm{lowrank}_d(P)-P\|_F^2/\|P\|_F^2$ of using a rank-$d$ approximation of $P$ for three different atlases.
The values of this curve indicate the increase in relative error when using $\mathrm{lowrank}_d(\bar{A})$ as compared to $\bar{A}$, when $M$ is large.
The relative error increases decay relatively slowly , indicating that $P$ is not well approximated by a low-rank matrix.
If $P$ were actually low-rank, 13, 8, and 37, respectively, the relative error is plotted with the dashed curves.}
\label{fig:screeplot}
\end{figure}

While these challenges can negatively impact the performance of low-rank procedures when estimating $P$, we can also view $\hat{P}$ as estimating latent low-rank structure in $P$.
For such an estimand, $\hat{P}$ will provide excellent performance, even for large $M$.

To illustrate this, we considered estimating the $d^*$ approximation of $P$, denoted as $P_{d^*}$, where $d^*$ is chosen according the Zhu and Ghodsi method applied to the eigenvalues of $P$.
For the SWU4 dataset with the CPAC 200 atlas, Figure~\ref{fig:latent} shows the estimated relative efficiency, on a log-scale, of $\bar{A}$ compared to $\hat{P}$, for estimating $P$, \replaced{red solid line}{left panel}, and $P_{d^*}$, \replaced{dashed blue line}{right panel}.
The relative efficiency is estimated based on 100 monte carlo replicates of sampling $M$ graphs for $M\in \{1, 2, 5, 10, 20, 50, 100\}$.
Recalling that relative efficiencies below one favor $\hat{P}$, for large $M$, $\hat{P}$ has poor performance compared to $\bar{A}$ for estimating $P$ but excellent relative efficiency for estimating $P_{d^*}$.

\begin{figure}[!htb]
    \centering
    \includegraphics[width=\linewidth]{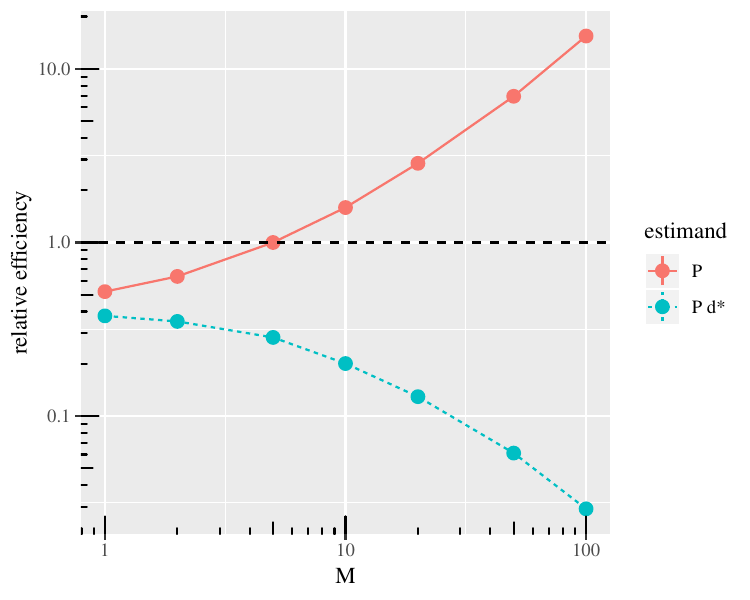}
    \caption{For the CPAC 200 atlas, the estimated relative efficiency (RE) of $\bar{A}$ compared to $\hat{P}$ on a log-scale, where values below one favor $\hat{P}$.
    The red solid line is the RE for estimating $P$ and the blue dashed line is the RE for estimating $P_{d^*}$
    The horizontal axis denotes the sample size $M$ on a log-scale.
    While $\hat{P}$ is a relatively poor estimator of $P$ for large $M$, the reverse is true for $P_{d^*}$, where $\hat{P}$ is clearly superior to $\bar{A}$, even for large $M$.}
    \label{fig:latent}
\end{figure}

This dataset has other challenges for low-rank methods.
First, there are a large number of negative eigenvalues which $\hat{P}$ will not capture.

Low-rank methods can include large negative eigenvalues, however, for low sample sizes excluding negative eigenvalues improved performance for SWU4.
See Appendix~\ref{app:compare_param} for a comparison of these and other parameters in the CPAC200 atlas.

Second, approximately 12.8\% of the entries of $P$ are exactly equal to 1.
For these edges, $\bar{A}$ will have exactly zero error, while $\hat{P}$ will be a less accurate estimate.



Despite these challenges, when the sample size is relatively small, such as $M=1$ or $M=5$, and for a larger number of vertices, $\hat{P}$ gives a better estimate than $\bar{A}$ for the SWU4 dataset.
(See Appendix~\ref{app:sim_iem} for a similar synthetic data analysis.)
Importantly, this improvement is robust to the embedding dimension as illustrated in Appendix~\ref{app:dim}.

\subsection{Eigen-connectomes and Lobe Structure}
\label{section:lobe_structure}

In addition to yielding potentially improved estimation, low-rank methods simultaneously provide convenient interpretations.

Using the mean graph $P$ for the Desikan atlas, the average of all 454 graphs, we estimated latent positions $\hat{X} \in \mathbb{R}^{N\times d}$.
Here, $N=70$ is the number of vertices and $d = 8$ is the dimension selected by the Zhu and Ghodsi's method \cite{zhu2006automatic}.
Fig.~\ref{fig:eigenvector_brain} shows the first 4 dimensions of $\hat{X}$ in the brain space. The value of $\hat{X}_{ij}$ determines the color of the $i$-th brain region for the $j$-th dimension; i.e.\ the $j$-th entry of the estimated latent vector for the $i$-th region. Red represents a positive value while blue represents a negative one, and the darker the color, the smaller the magnitude of the $\hat{X}_{ij}$.

The 1st dimension, depicted in Panel (a) of the figure, correlates strongly with degrees for each vertex. 
The 2nd dimension,  panel (b), shows  a distinction between the left and right hemisphere. Similarly, the other dimensions qualitatively correspond to different lobes.
For example, the red color corresponds to the Frontal and Temporal lobes in the 3rd dimension, while the light blue roughly matches the Occipital lobe in the 4th dimension. 

These ``eigen-connectomes'' demonstrate noteworthy similarity to the structural connectome harmonics of \cite{Atasoy2016-ip}.
While prior work has established the utility of eigen-connectomes as a basis set for modeling of functional dynamics, the current study demonstrates the correspondence of these dimensions with lobular divisions.
The connection between eigenvectors corresponding to each dimension and lobes will be explored more rigorously in Section~\ref{section:lobe_structure}.
Future work integrating the connections between structural features with the modeling of dynamics may provide insight into the spatial distribution of large-scale cortical hierarchies~\cite{Margulies2016-jj}.
Additionally, such representations enable the use of techniques from multivariate analysis to further study the mean graph.

\begin{figure}
\begin{subfigure}[t]{0.95\linewidth}
\caption{1st dimension}
\vspace{-7pt}
\centering
  \includegraphics[trim={5cm 5cm  4cm  4cm },clip,height=.3\linewidth]{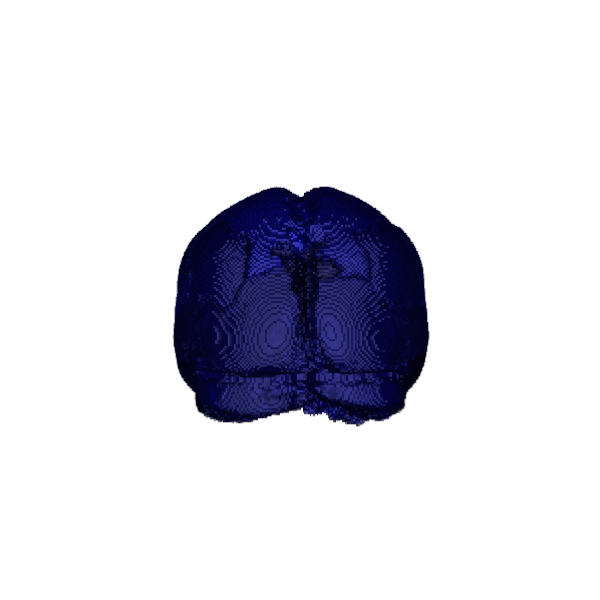}\hspace{-10pt}
  \includegraphics[trim={5cm 5cm  4cm  4cm },clip,height=.3\linewidth]{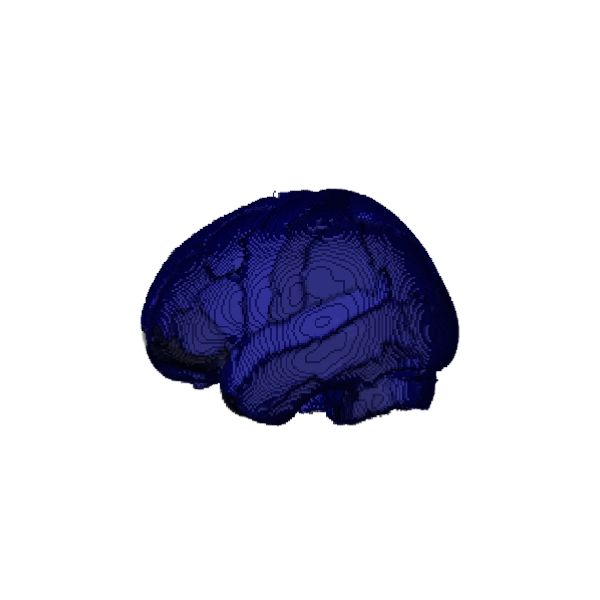}\hspace{-10pt}
  \includegraphics[trim={5cm 5cm  4cm  4cm },clip,height=.3\linewidth]{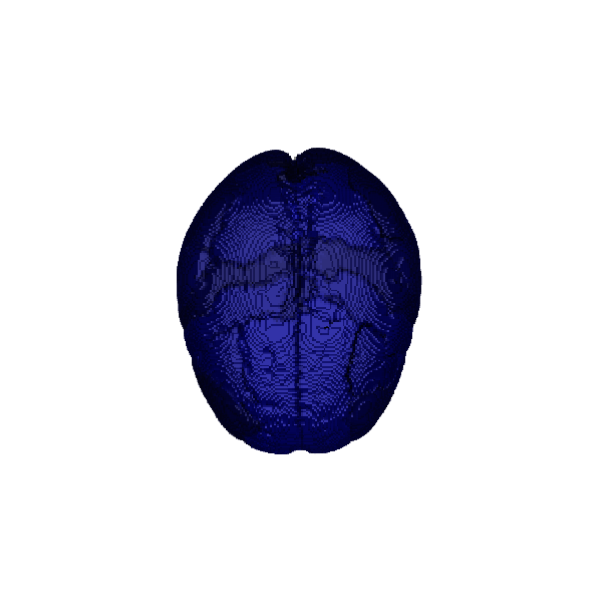}
\end{subfigure}
\begin{subfigure}[t]{0.95\linewidth}
\caption{2nd dimension}
\vspace{-7pt}
\centering
  \includegraphics[trim={5cm 5cm  4cm  4cm },clip,height=.3\linewidth]{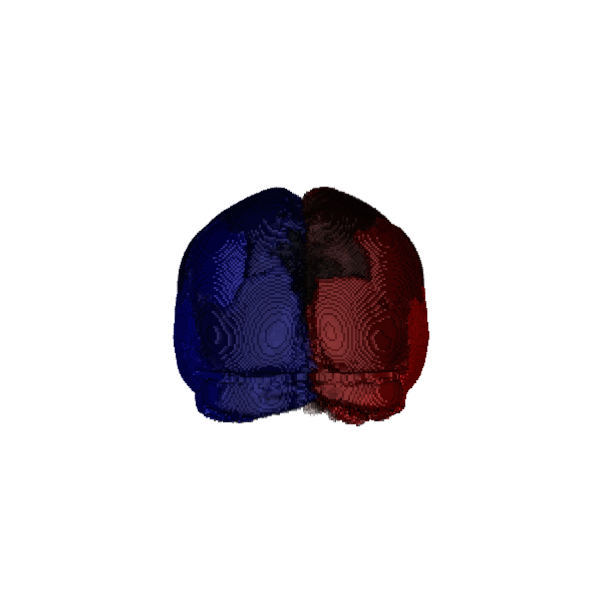}\hspace{-10pt}
  \includegraphics[trim={5cm 5cm  4cm  4cm },clip,height=.3\linewidth]{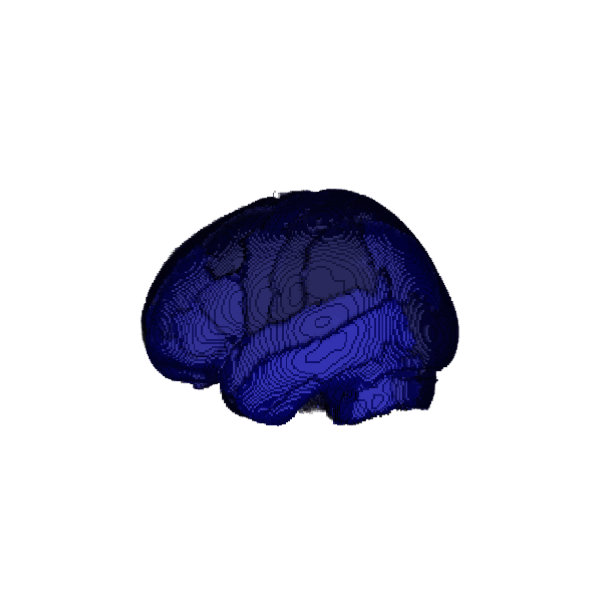}\hspace{-10pt}
  \includegraphics[trim={5cm 5cm  4cm  4cm },clip,height=.3\linewidth]{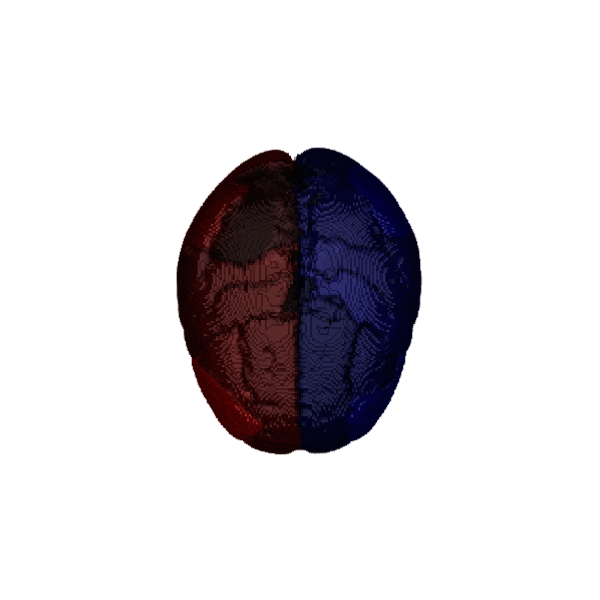}
\end{subfigure}
\begin{subfigure}[t]{0.95\linewidth}
\caption{3rd dimension}
\vspace{-7pt}
\centering
  \includegraphics[trim={5cm 5cm  4cm  4cm },clip,height=.3\linewidth]{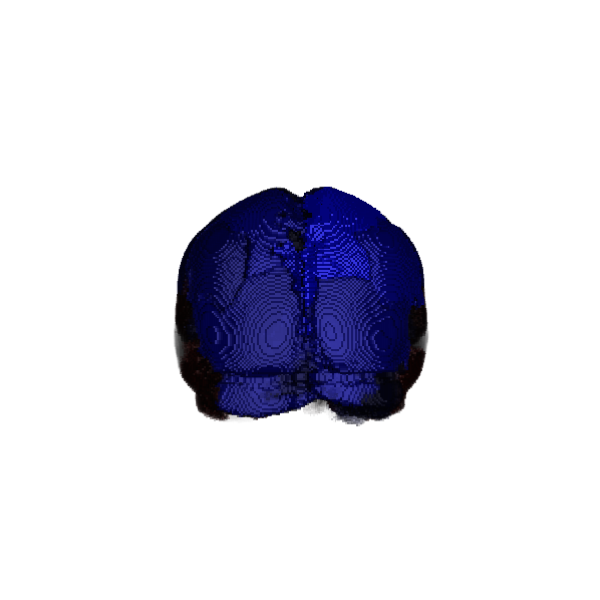}\hspace{-10pt}
  \includegraphics[trim={5cm 5cm  4cm  4cm },clip,height=.3\linewidth]{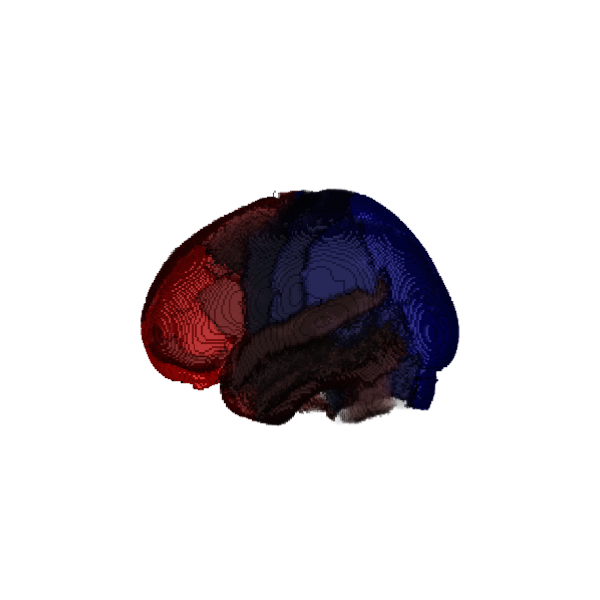}\hspace{-10pt}
  \includegraphics[trim={5cm 5cm  4cm  4cm },clip,height=.3\linewidth]{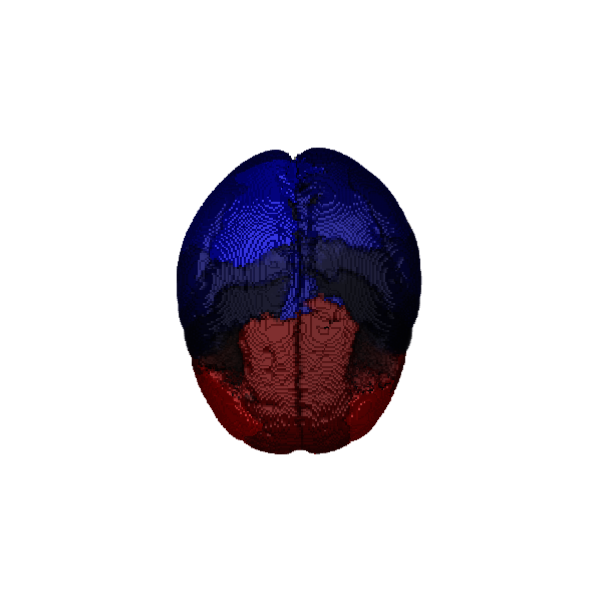}
\end{subfigure}
\begin{subfigure}[t]{0.95\linewidth}
\caption{4th dimension}
\vspace{-7pt}
\centering
  \includegraphics[trim={5cm 5cm  4cm  4cm },clip,height=.3\linewidth]{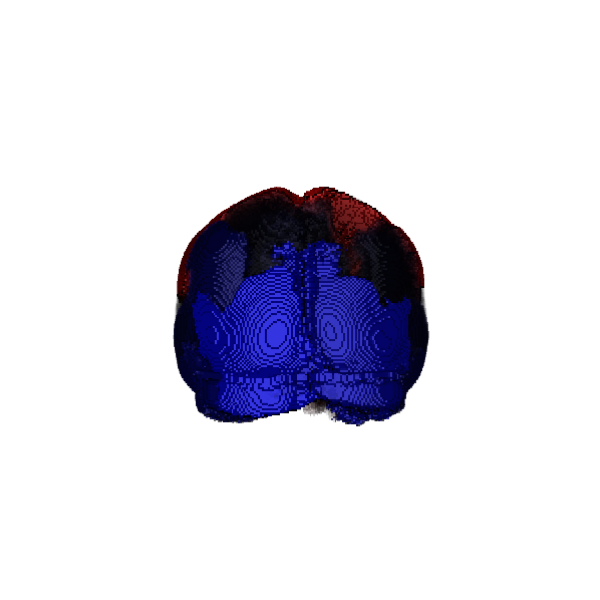}\hspace{-10pt}
  \includegraphics[trim={5cm 5cm  4cm  4cm },clip,height=.3\linewidth]{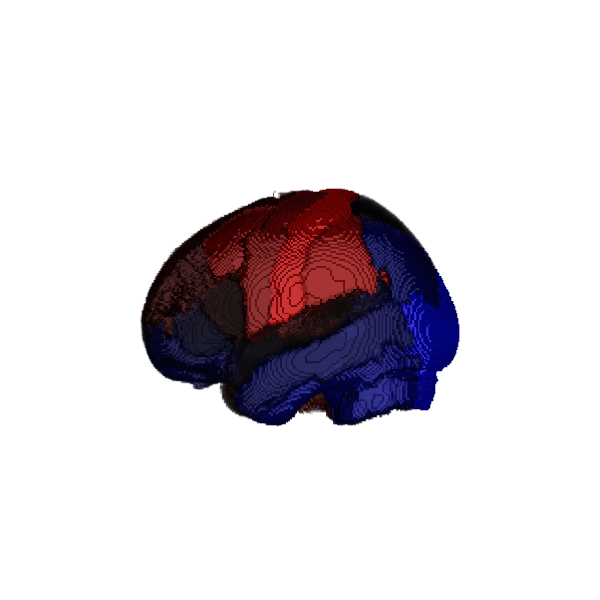}\hspace{-10pt}
  \includegraphics[trim={5cm 5cm  4cm  4cm },clip,height=.3\linewidth]{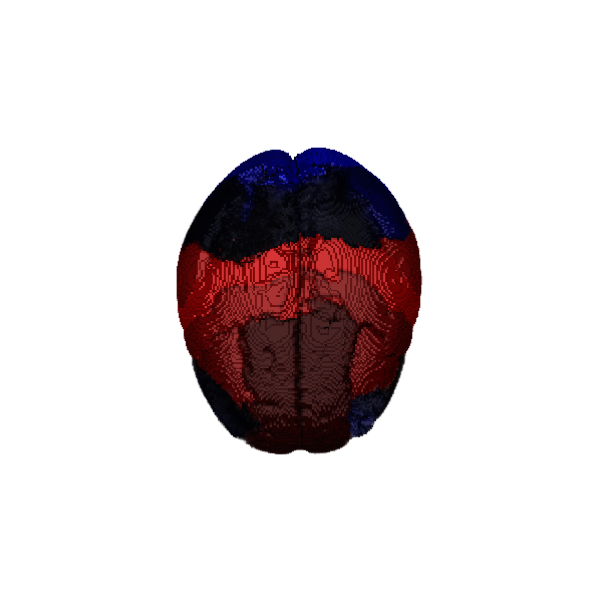}
\end{subfigure}
\caption{
Each panel depicts the values of the latent positions for the regions in the Desikan atlas for the mean graph of the SWU4 dataset.
The color of the $i$-th brain region for the $j$-th dimension is determined by the value of $\hat{X}_{ij}$, i.e.\ the $j$-th element of the estimated latent vector for the $i$-th region. Red represents a positive value while blue represents the negative one, with brighter color indicating larger magnitudes.
The 1st dimension, depicted in Panel (a), is relatively flat across the entire brain. 
In Panel (b), there is a distinction of the left and right hemispheres as conveyed in the 2nd dimension. Similarly, the other dimensions appear to correspond closely with the anatomical lobe structures of the brain.}
\label{fig:eigenvector_brain}
\end{figure}


As an example, the brain can be divided into lobes, originally based purely on anatomical considerations, now widely recognized to also play a functional role~\cite{Vanderah2015}.
While different anatomists partition brain regions differently, there is general agreement on four cortical lobes per hemisphere: frontal, parietal, occipital, and temporal \cite{fischl2012freesurfer,desikan2006automated,salat2004thinning}.

For the Desikan atlas, there are 70  regions (35 regions for each hemisphere), with each region belonging to a single lobe (see Appendix~\ref{app:data_human}). 

One might hypothesize that properties of regions within a lobe are more similar than across lobes, as regions within lobes would be expected to share more functional roles.
To test whether the embedded latent positions $X$ preserve this property or not, we propose a test statistic $T$ to be the average differences between vertices within the same lobe minus the average differences between vertices across different lobes, i.e.

\begin{equation}
	T(X, l) = \sum_{i \ne j} \frac{\mathbf{1}_{l(i) = l(j)}\|X_i - X_j \|_2}{\sum_{k \ne l} \mathbf{1}_{l(k) = l(l)}} -
\frac{ \mathbf{1}_{l(i) \ne l(j)}\|X_i - X_j \|_2}{\sum_{k \ne l} \mathbf{1}_{l(k) \ne l(l)}}, \label{eq:test_stat}
\end{equation}
where $l(i)$ denotes the lobe assignment for vertex $i$.
If the latent positions $X$ and the lobe assignment $l$ are independent, then $T(X, l)$ will be close to zero.
A small test statistic $T(X, l)$ indicates that latent positions of the regions within the same lobe are closer compared to the ones across the lobes.


However, the anatomical geometry might contribute to the dependence between $X$ and $l$, with spatially proximal vertices having similar connectivity patterns.
Hence, a small test statistic $T(X, l)$ is evidence that the low-rank methods preserve the lobe structure only if we also condition on anatomy geometry:
$H_0: X$ and $l$ are conditionally independent given anatomical geometry,
$H_A: X$ and $l$ are conditionally dependent given anatomical geometry.
The test of conditional independence has less power compared to the test of unconditional independence which is performed with a random permutation of lobe labels (which yield a p-value $<10^{-6}$).

To test under the anatomical geometry conditions, the lobe assignments $l(i)$ were randomly modified so that the number of regions in each lobe remain the same and the regions within the same modified lobe are still spatially connected. 
In particular, we performed a sequence of randomized flips, where 
a flip is a swap of two pairs of vertices which preserves the number of regions in each lobe and  maintains the constraint that lobes are spatially contiguous.
The lobes are flipped a limited number of times in order to study how the number of flips impacts the test statistic.
Appendix~\ref{app:testing} discusses the flipping procedure.

1000 simulations with the test statistics  of $T(X, l')$ are performed, each with a fixed number of flips. 
The number of flips varies from 1 to 10  (Fig.~\ref{fig:violin_plot}). In the violin plot, the dashed line indicates the value of $T(X, l)$ based on the true lobe assignment. 
The p-value is less than 0.05 if the number of flips is larger than 7.
Hence, latent positions in the same lobe are more similar to each other, even after accounting for the fact that geometrically proximal regions may also have similar latent positions.

When the number of flips is small, this test has very little power, with the null distribution being only a small deviation from the original lobes.
When the number of flips gets large, eventually the contiguity of the lobes breaks down and the empirical $p$-values continue to get smaller.

\begin{figure}[!htbp]
\centering
\includegraphics[width=1\linewidth]{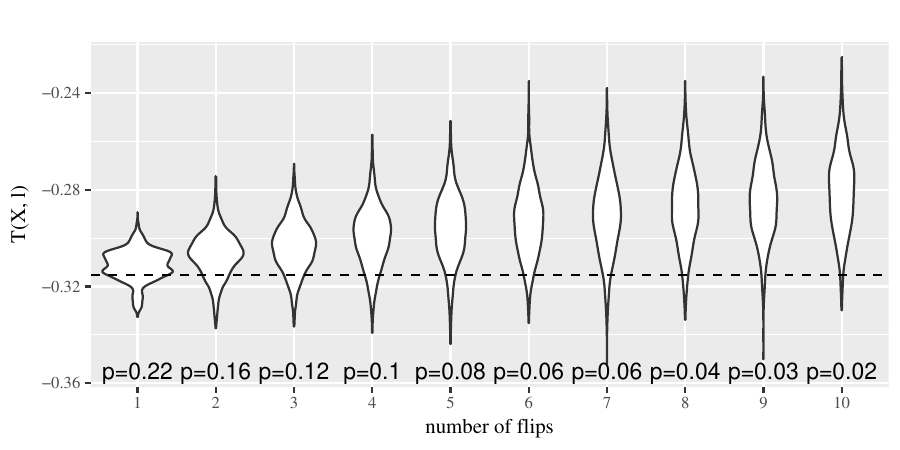}
\caption{
Violin plot for 1000 simulations of $T(X,l')$ with $l'$ randomly drawn from the set of valid flips, preserving lobe size and contiguity. The dashed line represents the statistic $T(x,l)$ for the true lobes. As the null hypotheses move further from the original lobe assignment, with the number of flips growing, the 95\% central region of the null shifts away from $T(x,l)$ under the true lobe assignment.
The $p$-value for each test is provided at the bottom of the figure.}
\label{fig:violin_plot}
\end{figure}

\section{Application to a Mouse Connectome}

As a further application of low-rank methods, an MRI-DTI mouse brain connectome \cite{Calabrese2015} with M=1 specimen was evaluated.
The data acquisition protocol is described in the appendix, Appendix~\ref{app:data_mouse}, and resulted in a $296$ node\ weighted, directed graphs with vertices again corresponding to regions in the brain. 
The  296 regions were organized into a multilevel, hierarchical structure. Analysis of the fine-grained and the first level of the hierarchy partitioned the label set into eight superstructures, with four in each hemisphere: forebrain, midbrain, hindbrain, and white matter. 

The original matrix $W\in (\Re^+)^{296\times 296}$ is a weighted adjacency matrix with $W_{ij}$ denoting the number of tracts passing through ROIs $i$ and $j$.
As the original weights had very heavy tails, these weights were transformed by setting  $A_{ij}=\log(W_{ij}+1)$.
This resulted in the weighted adjacency matrix in Figure~\ref{fig:mouse_adj_lr} (a).
Note that Algorithm~\ref{algo:basic} does not strictly require the entries to be binary and hence is applicable in this weighted setting by only thresholding the elements of $\hat{P}$ to be non-negative.

\begin{figure}[tbh!]
	\centering
	\includegraphics[width=\linewidth]{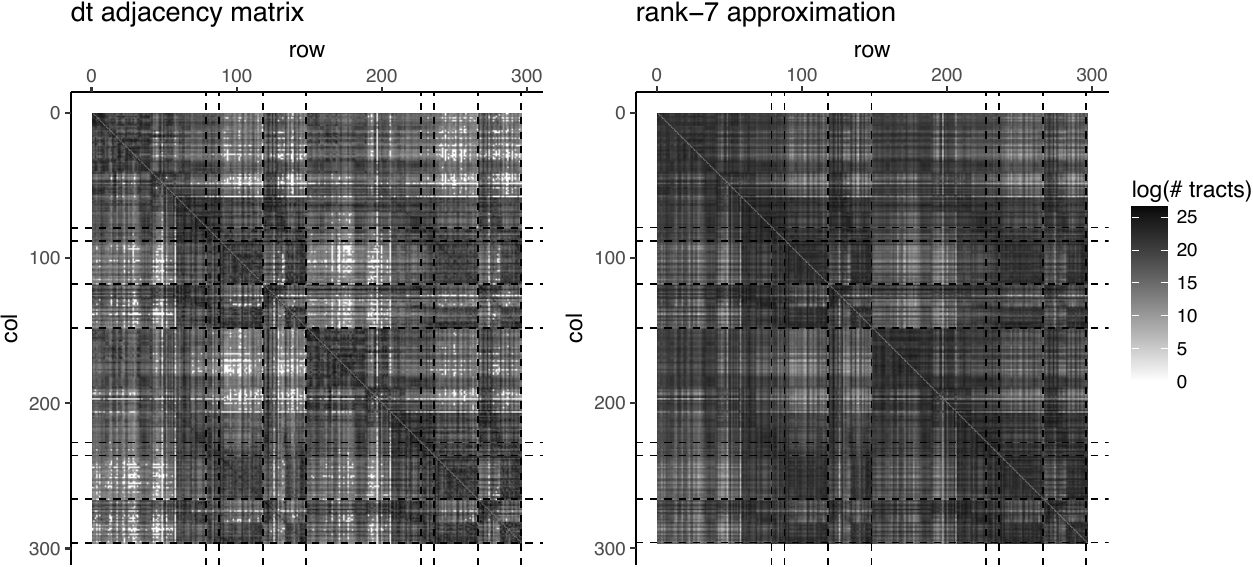}
	\caption{The left panel shows the weighted adjacency matrix with weights transformed by the transformation $w\mapsto \log(w+1)$.
	Higher weights are shown as darker pixels.
	The right panel shows the rank-7 approximation of this matrix where the rank was chosen using method in \cite{zhu2006automatic}.
	In both panels, dashed lines show the division between the eight different superstructures. }
	\label{fig:mouse_adj_lr}
\end{figure}

Using the procedure described in Algorithm~\ref{algo:basic}, the dimension selection procedure \cite{zhu2006automatic} resulted in a rank-7 approximation which is shown in the right panel of Figure~\ref{fig:mouse_adj_lr}.
Since the sample size is only one, cross-validation cannot be employed, but visually it appears that the rank-7 approximation captures many of the features in the original matrix.

As with the human data, we studied the relationship between the structure of the graph and the eight superstructures.
Panel (a) of Figure~\ref{fig:mouse_ase} shows entries grouped by superstructure of the four scaled singular vectors  of the weighted adjacency matrix $A$ corresponding to the largest singular values. 
The points are colored according to the four superstructures and the shapes are determined by the hemisphere. 
The ordering of the points groups together nodes in the same superstructure and hemisphere.
The second and fourth vectors have structure which correlates closely with the four superstructures and the two hemispheres, respectively.
Additionally, the first vector appears to separate the midbrain from the other three superstructures.

The right panel of Figure~\ref{fig:mouse_ase} shows a scatter plot of the entries of the fourth and second vectors along with the class boundaries for the eight-class quadratic discriminant analysis classifier. 
This classifier achieves a training error rate of $87/296\approx 0.29$.
The error rate is particularly high for the white matter, with $58/60$ vertices being classified incorrectly, meaning that ignoring the white matter, $29/236 \approx 0.12$ vertices were misclassified. 
Panel (c) shows the normalized confusion matrix for the eight classes, indicating that the forebrain and hindbrain classes are well separated while the white matter and midbrain have more substantial overlap.
This matches with the general structure of the white matter which is not defined at the first hierarchical level of the atlas according to spatial structure, while the fore-, mid-, and hind-brain superstructures are.

Finally, the same permutation analysis as Fig.~\ref{fig:violin_plot} was also performed for the mouse connectome as shown in Fig.~\ref{fig:mouse_violin} in Appendix~\ref{app:testing}.
This test again indicates that the eight superstructures are significant even after accounting for the spatial structure of the regions.

\begin{figure*}[tbh!]
	\centering

	\includegraphics[width=.90\linewidth]{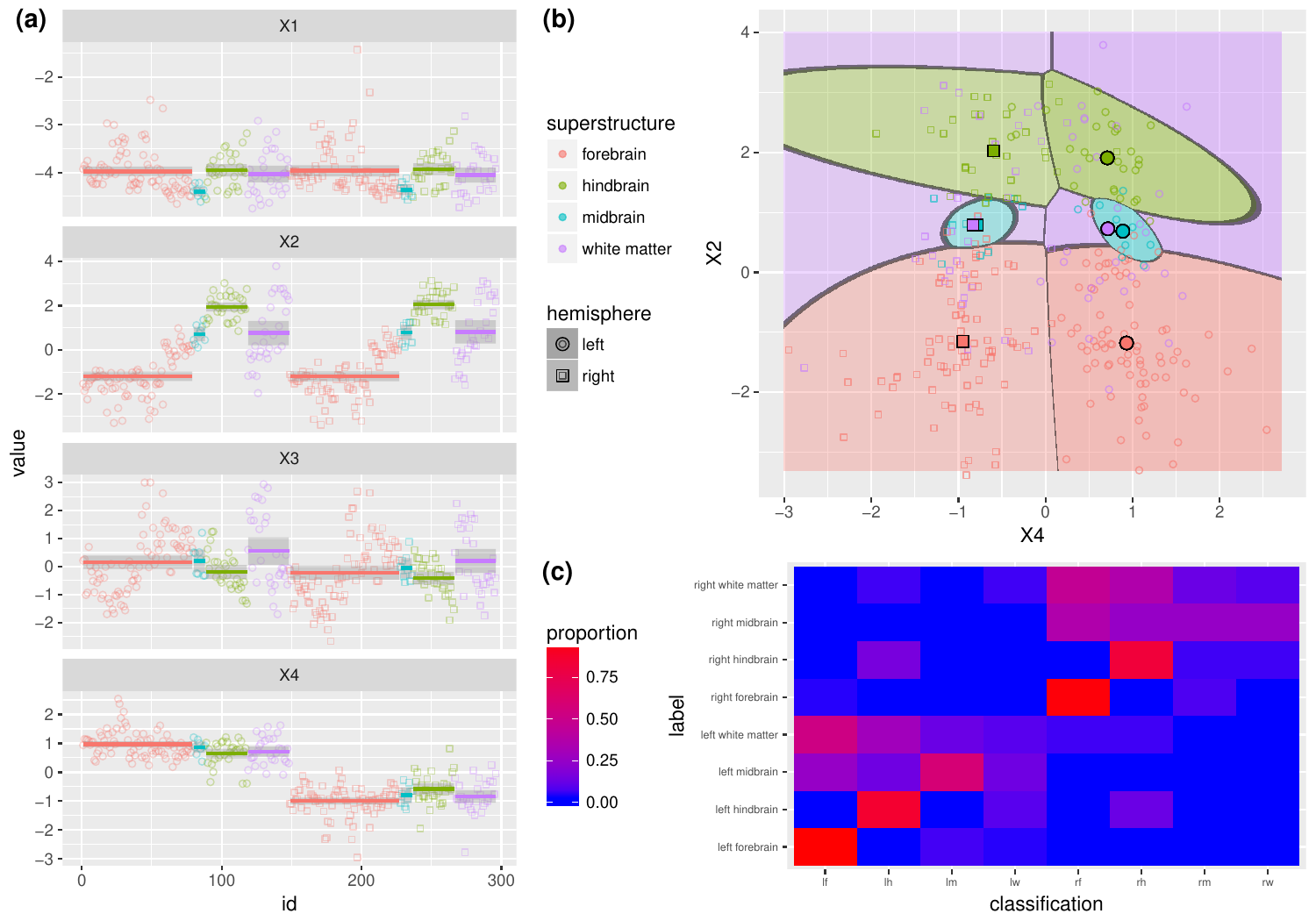}
	\caption{
	Panel (a) shows entries of the the 4 scaled singular vectors  of $A$ corresponding to largest singular values.
	The ordering of the entries is the same as the ordering in Figure~\ref{fig:mouse_adj_lr}, where vertices in the same superstructures and hemispheres are grouped together. 
	Colors of points denote the superstructures and shapes denote the hemispheres.
	Panel (b) is a scatter plot of the the fourth (horizontal axis) and second (vertical axis) singular vectors. 
	The background coloring is determined by learning a classifier, based on a mixture of Gaussians, to classify the eight different superstructures.
	Panel (c) shows the normalized confusion matrix for the mixture of Gaussians learned in panel (b).
    Each entry corresponds to the proportion of nodes with a given label that were predicted to be each other label,
    with true and predicted labels given by the row and column, respectively.}
	\label{fig:mouse_ase}
\end{figure*}


\section{Discussion}\label{sec:discussion}

Motivated by the RDPG model, via a low-rank approximation to the entry-wise MLE, our methodology takes advantage of low-rank structure of graphs.
We give a closed form for the asymptotic relative efficiency between the entry-wise MLE $\bar{A}$ and \replaced{its low-rank approximation $\mathrm{lowrank}_d(\bar{A})$}{our estimator $\hat{P}$} in the case of a SBM, showing theoretically that low-rank methods can provide substantial improvements.
\added{Our theoretical results are all shown for the large $N$ regime in terms of $\mathrm{lowrank}_d(\bar{A})$, however in practice we have observed that our proposed estimator $\hat{P}$ has even better performance, and asymptotically the difference between them will typically be small under dense low-rank models.}
Moreover, our estimator outperforms the entry-wise MLE in a cross validation analysis of the SWU4 brain graphs and in low- and full-rank simulation settings when $M$ is small.
These results illustrate that $\hat{P}$ performs well even when the low-rank assumption is violated and that $\hat{P}$ is robust and can be applied in practice.

The low-rank methods could also be applied for functional MRI studies by applying them to appropriate correlation matrices or other network estimates.
We have performed preliminary explorations in that setting where $\hat{P}$ again performs well for estimating latent structure.
Since static fMRI graphs likely have lower noise than structural connectomes, the estimator $\bar{A}$ generally performed well for estimating the population mean graph.

As we have shown in Figure~\ref{fig:latent}, $\hat{P}$ is an excellent estimate of the low-rank latent structure of $P$.
As other authors have noted~\cite{Udell2016-yj}, many random network models will enjoy strong low-rank structure in their mean graph which will be readily captured by $\hat{P}$.

For the human connectome data, the largest improvements using the low-rank method occurred when the number of graphs $M$ was small, while it provided only minor improvements, or even slightly degraded performance, when $M$ was large.
However, even in large scale studies, low-rank methods will be useful for estimating graph means for subpopulations, e.g.\ the population of females over 60 with some college education.
Using the element-wise sample mean for such small strata, which may have fewer than ten subjects, will frequently result in a degradation of performance.
Similarly, \cite{durante2016nonparametric} proposed a Bayesian nonparametric approach for modeling the population distribution of network-valued data which reduces dimensionality via a mixture model and our methods could be easily adapted to those ideas.

While the low-rank methods considered in this paper may perform well, further refinements of these methods which account for the particular traits of connectomics data would be useful to improve estimation further.
For example, we assume that the adjacency matrix is observed without contamination.
However, when heavy-tailed noise is present, robust methods may be necessary.
Rank-based methods and robust likelihood methods could be very useful in that case \cite{huber2009robust,qin2013maximum}.
\cite{Tang2017-yv} considers this problem and proposes the use of $L^q$ likelihood to improve both the low-rank and full rank methods for estimating $P$.


Another issue that arose in the analysis of the connectome dataset was the presence of structural ones in the mean graph for the population.
These structural ones appear since edges between certain regions of the brain are present in all members of the healthy population.
For these always-present edges, the low-rank methods will have non-zero error while the sample mean will always have zero error.
Detecting and incorporating structural ones and zeros could yield methods that share the best elements of both methods considered here.


While in this paper the focus is the estimation of the mean graph $P$ exploiting the low-rank structure, many future directions are quite interesting, such as fitting the SBM or clustering the vertices to detect different brain regions.
For example, \cite{abraham2013extracting} introduced a region-extraction approach based on a sparse penalty with dictionary learning; \cite{calhoun2001method} performed independent component analysis of fMRI data to draw group inferences; \cite{wang2017joint} consider a joint embedding model for feature extractions.





\bibliographystyle{IEEEtran}

\bibliography{Bib.bib}

\begin{thebibliography}{10}
\providecommand{\url}[1]{#1}
\csname url@samestyle\endcsname
\providecommand{\newblock}{\relax}
\providecommand{\bibinfo}[2]{#2}
\providecommand{\BIBentrySTDinterwordspacing}{\spaceskip=0pt\relax}
\providecommand{\BIBentryALTinterwordstretchfactor}{4}
\providecommand{\BIBentryALTinterwordspacing}{\spaceskip=\fontdimen2\font plus
\BIBentryALTinterwordstretchfactor\fontdimen3\font minus
  \fontdimen4\font\relax}
\providecommand{\BIBforeignlanguage}[2]{{%
\expandafter\ifx\csname l@#1\endcsname\relax
\typeout{** WARNING: IEEEtran.bst: No hyphenation pattern has been}%
\typeout{** loaded for the language `#1'. Using the pattern for}%
\typeout{** the default language instead.}%
\else
\language=\csname l@#1\endcsname
\fi
#2}}
\providecommand{\BIBdecl}{\relax}
\BIBdecl

\bibitem{trunk1979problem}
G.~V. Trunk, ``A problem of dimensionality: A simple example,'' \emph{Pattern
  Analysis and Machine Intelligence, IEEE Transactions on}, no.~3, pp.
  306--307, 1979.

\bibitem{ginestet2014hypothesis}
C.~E. Ginestet, P.~Balanchandran, S.~Rosenberg, and E.~D. Kolaczyk,
  ``Hypothesis testing for network data in functional neuroimaging,''
  \emph{arXiv preprint arXiv:1407.5525}, 2014.

\bibitem{Bhattacharyya2018-wj}
S.~Bhattacharyya and S.~Chatterjee, ``Spectral clustering for multiple sparse
  networks: {I},'' May 2018.

\bibitem{Udell2016-yj}
M.~Udell, C.~Horn, R.~Zadeh, and S.~Boyd, ``Generalized low rank models,''
  \emph{Foundations and Trends\textregistered{} in Machine Learning}, vol.~9,
  no.~1, pp. 1--118, 2016.

\bibitem{Eckart1936-tt}
C.~Eckart and G.~Young, ``\BIBforeignlanguage{en}{The approximation of one
  matrix by another of lower rank},''
  \emph{\BIBforeignlanguage{en}{Psychometrika}}, vol.~1, no.~3, pp. 211--218,
  Sep. 1936.

\bibitem{Berry2007-ok}
M.~W. Berry, M.~Browne, A.~N. Langville, V.~P. Pauca, and R.~J. Plemmons,
  ``Algorithms and applications for approximate nonnegative matrix
  factorization,'' \emph{Computational statistics \& data analysis}, vol.~52,
  no.~1, pp. 155--173, Sep. 2007.

\bibitem{Boutsidis2008-mx}
C.~Boutsidis and E.~Gallopoulos, ``{SVD} based initialization: A head start for
  nonnegative matrix factorization,'' \emph{Pattern recognition}, vol.~41,
  no.~4, pp. 1350--1362, Apr. 2008.

\bibitem{Lee2014-vm}
N.~H. Lee, I.-J. Wang, Y.~Park, C.~E. Priebe, and M.~Rosen, ``Automatic
  dimension selection for a non-negative factorization approach to clustering
  multiple random graphs,'' \emph{Stat}, vol. 1050, p.~24, 2014.

\bibitem{Banerjee2007-mk}
A.~Banerjee and J.~Jost, ``\BIBforeignlanguage{en}{Spectral plots and the
  representation and interpretation of biological data},''
  \emph{\BIBforeignlanguage{en}{Theory in biosciences = Theorie in den
  Biowissenschaften}}, vol. 126, no.~1, pp. 15--21, Aug. 2007.

\bibitem{Rahim2017-ci}
M.~Rahim, B.~Thirion, and G.~Varoquaux, ``{Population-Shrinkage} of covariance
  to estimate better brain functional connectivity,'' in \emph{Medical Image
  Computing and Computer Assisted Intervention − {MICCAI} 2017}.\hskip 1em
  plus 0.5em minus 0.4em\relax Springer International Publishing, 2017, pp.
  460--468.

\bibitem{Atasoy2016-ip}
S.~Atasoy, I.~Donnelly, and J.~Pearson, ``\BIBforeignlanguage{en}{Human brain
  networks function in connectome-specific harmonic waves},''
  \emph{\BIBforeignlanguage{en}{Nature communications}}, vol.~7, p. 10340, Jan.
  2016.

\bibitem{Margulies2016-jj}
D.~S. Margulies, S.~S. Ghosh, A.~Goulas, M.~Falkiewicz, J.~M. Huntenburg,
  G.~Langs, G.~Bezgin, S.~B. Eickhoff, F.~X. Castellanos, M.~Petrides,
  E.~Jefferies, and J.~Smallwood, ``\BIBforeignlanguage{en}{Situating the
  default-mode network along a principal gradient of macroscale cortical
  organization},'' \emph{\BIBforeignlanguage{en}{Proceedings of the National
  Academy of Sciences of the United States of America}}, vol. 113, no.~44, pp.
  12\,574--12\,579, Nov. 2016.

\bibitem{De_Lange2014-hd}
S.~C. de~Lange, M.~A. de~Reus, and M.~P. van~den Heuvel,
  ``\BIBforeignlanguage{en}{The laplacian spectrum of neural networks},''
  \emph{\BIBforeignlanguage{en}{Frontiers in computational neuroscience}},
  vol.~7, p. 189, Jan. 2014.

\bibitem{gray2012magnetic}
W.~Gray, J.~Bogovic, J.~Vogelstein, B.~Landman, J.~Prince, and R.~Vogelstein,
  ``Magnetic resonance connectome automated pipeline: An overview,'' \emph{IEEE
  Pulse}, vol.~2, no.~3, pp. 42--48, 2012.

\bibitem{bollobas2007phase}
B.~Bollob{\'a}s, S.~Janson, and O.~Riordan, ``The phase transition in
  inhomogeneous random graphs,'' \emph{Random Structures \& Algorithms},
  vol.~31, no.~1, pp. 3--122, 2007.

\bibitem{Gilbert1959-ba}
E.~N. Gilbert, ``Random graphs,'' \emph{Annals of Mathematical Statistics},
  vol.~30, no.~4, pp. 1141--1144, 1959.

\bibitem{Erdos1959-ln}
P.~Erd{\H o}s and A.~R{\'e}nyi, ``{On random graphs, I},'' \emph{Publicationes
  Mathematicae Debrecen}, vol.~6, pp. 290--297, 1959.

\bibitem{young2007random}
S.~J. Young and E.~R. Scheinerman, ``Random dot product graph models for social
  networks,'' in \emph{Algorithms and models for the web-graph}.\hskip 1em plus
  0.5em minus 0.4em\relax Springer, 2007, pp. 138--149.

\bibitem{nickel2007random}
C.~L.~M. Nickel, ``Random dot product graphs a model for social networks,''
  Ph.D. dissertation, Johns Hopkins University, 2008.

\bibitem{hoff2002latent}
P.~D. Hoff, A.~E. Raftery, and M.~S. Handcock, ``Latent space approaches to
  social network analysis,'' \emph{Journal of the american Statistical
  association}, vol.~97, no. 460, pp. 1090--1098, 2002.

\bibitem{holland1983stochastic}
P.~W. Holland, K.~B. Laskey, and S.~Leinhardt, ``Stochastic blockmodels: First
  steps,'' \emph{Social networks}, vol.~5, no.~2, pp. 109--137, 1983.

\bibitem{airoldi2008mixed}
E.~M. Airoldi, D.~M. Blei, S.~E. Fienberg, and E.~P. Xing, ``Mixed membership
  stochastic blockmodels,'' \emph{Journal of Machine Learning Research},
  vol.~9, no. Sep, pp. 1981--2014, 2008.

\bibitem{karrer2011stochastic}
B.~Karrer and M.~E. Newman, ``Stochastic blockmodels and community structure in
  networks,'' \emph{Physical Review E}, vol.~83, no.~1, p. 016107, 2011.

\bibitem{Lyzinski2014-az}
V.~Lyzinski, D.~L. Sussman, M.~Tang, A.~Athreya, and C.~E. Priebe, ``Perfect
  clustering for stochastic blockmodel graphs via adjacency spectral
  embedding,'' \emph{Electronic journal of statistics}, vol.~8, no.~2, pp.
  2905--2922, 2014.

\bibitem{Rubin-Delanchy2017-av}
P.~Rubin-Delanchy, C.~E. Priebe, and M.~Tang, ``Consistency of adjacency
  spectral embedding for the mixed membership stochastic blockmodel,'' May
  2017.

\bibitem{chatterjee2015matrix}
S.~Chatterjee, ``Matrix estimation by universal singular value thresholding,''
  \emph{The Annals of Statistics}, vol.~43, no.~1, pp. 177--214, 2015.

\bibitem{marchette2011vertex}
D.~Marchette, C.~Priebe, and G.~Coppersmith, ``Vertex nomination via attributed
  random dot product graphs,'' in \emph{Proceedings of the 57th ISI World
  Statistics Congress}, vol.~6, 2011, p.~16.

\bibitem{scheinerman2010modeling}
E.~R. Scheinerman and K.~Tucker, ``Modeling graphs using dot product
  representations,'' \emph{Computational Statistics}, vol.~25, no.~1, pp.
  1--16, 2010.

\bibitem{sussman2014consistent}
D.~L. Sussman, M.~Tang, and C.~E. Priebe, ``Consistent latent position
  estimation and vertex classification for random dot product graphs,''
  \emph{IEEE transactions on pattern analysis and machine intelligence},
  vol.~36, no.~1, pp. 48--57, 2014.

\bibitem{zhu2006automatic}
M.~Zhu and A.~Ghodsi, ``Automatic dimensionality selection from the scree plot
  via the use of profile likelihood,'' \emph{Computational Statistics \& Data
  Analysis}, vol.~51, no.~2, pp. 918--930, 2006.

\bibitem{athreya2016limit}
A.~Athreya, C.~E. Priebe, M.~Tang, V.~Lyzinski, D.~J. Marchette, and D.~L.
  Sussman, ``A limit theorem for scaled eigenvectors of random dot product
  graphs,'' \emph{Sankhya A}, vol.~78, no.~1, pp. 1--18, 2016.

\bibitem{Fishkind2012}
D.~E. Fishkind, D.~L. Sussman, M.~Tang, J.~T. Vogelstein, and C.~E. Priebe,
  ``{Consistent adjacency-spectral partitioning for the stochastic block model
  when the model parameters are unknown},'' 2012.

\bibitem{Calabrese2015}
E.~Calabrese, A.~Badea, G.~Cofer, Y.~Qi, and G.~A. Johnson,
  ``\BIBforeignlanguage{en}{A diffusion {MRI} tractography connectome of the
  mouse brain and comparison with neuronal tracer data},''
  \emph{\BIBforeignlanguage{en}{Cereb. Cortex}}, vol.~25, no.~11, pp.
  4628--4637, Nov. 2015.

\bibitem{Tang2016-rs}
M.~Tang and C.~E. Priebe, ``Limit theorems for eigenvectors of the normalized
  laplacian for random graphs,'' \emph{arXiv preprint arXiv:1607. 08601}, 2016.

\bibitem{zuo2014open}
X.-N. Zuo, J.~S. Anderson, P.~Bellec, R.~M. Birn, B.~B. Biswal, J.~Blautzik,
  J.~C. Breitner, R.~L. Buckner, V.~D. Calhoun, F.~X. Castellanos
  \emph{et~al.}, ``An open science resource for establishing reliability and
  reproducibility in functional connectomics,'' \emph{Scientific data}, vol.~1,
  2014.

\bibitem{kiar2017science}
G.~Kiar, K.~J. Gorgolewski, D.~Kleissas, W.~G. Roncal, B.~Litt, B.~Wandell,
  R.~A. Poldrack, M.~Wiener, R.~J. Vogelstein, R.~Burns \emph{et~al.},
  ``Science in the cloud (sic): A use case in mri connectomics,'' \emph{Giga
  Science}, vol.~6, no.~5, pp. 1--10, 2017.

\bibitem{kiar2016ndmg}
G.~Kiar, W.~Gray~Roncal, D.~Mhembere, E.~Bridgeford, R.~Burns, and J.~T.
  Vogelstein, ``ndmg: Neurodata's mri graphs pipeline,'' \url{http://m2g.io},
  Aug. 2016.

\bibitem{oishi2010mri}
K.~Oishi, A.~V. Faria, P.~C. van Zijl, and S.~Mori, \emph{MRI atlas of human
  white matter}.\hskip 1em plus 0.5em minus 0.4em\relax Academic Press, 2010.

\bibitem{desikan2006automated}
R.~S. Desikan, F.~S{\'e}gonne, B.~Fischl, B.~T. Quinn, B.~C. Dickerson,
  D.~Blacker, R.~L. Buckner, A.~M. Dale, R.~P. Maguire, B.~T. Hyman
  \emph{et~al.}, ``An automated labeling system for subdividing the human
  cerebral cortex on mri scans into gyral based regions of interest,''
  \emph{Neuroimage}, vol.~31, no.~3, pp. 968--980, 2006.

\bibitem{sikka2014towards}
S.~Sikka, B.~Cheung, R.~Khanuja, S.~Ghosh, C.~Yan, Q.~Li, J.~Vogelstein,
  R.~Burns, S.~Colcombe, C.~Craddock \emph{et~al.}, ``Towards automated
  analysis of connectomes: The configurable pipeline for the analysis of
  connectomes (c-pac),'' in \emph{5th INCF Congress of Neuroinformatics,
  Munich, Germany}, vol.~10, 2014.

\bibitem{Vanderah2015}
T.~Vanderah and D.~J. Gould, \emph{\BIBforeignlanguage{en}{Nolte's The Human
  Brain: An Introduction to its Functional Anatomy, 7e}}, 7th~ed.\hskip 1em
  plus 0.5em minus 0.4em\relax Elsevier, 1~Jun. 2015.

\bibitem{fischl2012freesurfer}
B.~Fischl, ``Freesurfer,'' \emph{Neuroimage}, vol.~62, no.~2, pp. 774--781,
  2012.

\bibitem{salat2004thinning}
D.~H. Salat, R.~L. Buckner, A.~Z. Snyder, D.~N. Greve, R.~S. Desikan, E.~Busa,
  J.~C. Morris, A.~M. Dale, and B.~Fischl, ``Thinning of the cerebral cortex in
  aging,'' \emph{Cerebral cortex}, vol.~14, no.~7, pp. 721--730, 2004.

\bibitem{durante2016nonparametric}
D.~Durante, D.~B. Dunson, and J.~T. Vogelstein, ``Nonparametric bayes modeling
  of populations of networks,'' \emph{Journal of the American Statistical
  Association}, no. just-accepted, 2016.

\bibitem{huber2009robust}
P.~J. Huber and E.~M. Ronchetti, \emph{Robust statistics}, 2nd~ed., ser. Wiley
  Series in Probability and Statistics.\hskip 1em plus 0.5em minus 0.4em\relax
  John Wiley \& Sons, Inc., Hoboken, NJ, 2009.

\bibitem{qin2013maximum}
Y.~Qin and C.~E. Priebe, ``Maximum {$L_q$}-likelihood estimation via the
  expectation-maximization algorithm: A robust estimation of mixture models,''
  \emph{Journal of the American Statistical Association}, vol. 108, no. 503,
  pp. 914--928, 2013.

\bibitem{Tang2017-yv}
R.~Tang, M.~Tang, J.~T. Vogelstein, and C.~E. Priebe, ``Robust estimation from
  multiple graphs under gross error contamination,'' Jul. 2017.

\bibitem{abraham2013extracting}
A.~Abraham, E.~Dohmatob, B.~Thirion, D.~Samaras, and G.~Varoquaux, ``Extracting
  brain regions from rest fmri with total-variation constrained dictionary
  learning,'' in \emph{MICCAI-16th International Conference on Medical Image
  Computing and Computer Assisted Intervention-2013}.\hskip 1em plus 0.5em
  minus 0.4em\relax Springer, 2013.

\bibitem{calhoun2001method}
V.~D. Calhoun, T.~Adali, G.~D. Pearlson, and J.~Pekar, ``A method for making
  group inferences from functional mri data using independent component
  analysis,'' \emph{Human brain mapping}, vol.~14, no.~3, pp. 140--151, 2001.

\bibitem{wang2017joint}
S.~Wang, J.~T. Vogelstein, and C.~E. Priebe, ``Joint embedding of graphs,''
  \emph{arXiv preprint arXiv:1703.03862}, 2017.

\end{thebibliography}

\clearpage
\pagenumbering{arabic}
\setcounter{page}{1}
\appendix

\section{Methods}
\label{app:method}

\subsection{Choosing Dimension}
\label{app:dim_select}
Often in dimensionality reduction techniques, the choice for dimension $d$, relies on analyzing the set of the ordered eigenvalues, looking for a ``gap'' or ``elbow'' in the scree-plot. \citeapp{zhu2006automatic} present an automated method for finding this gap in the scree-plot that takes only the ordered eigenvalues as an input and uses Gaussian mixture modeling to find these gaps.
The mixture modeling results in multiple candidate dimensions or elbows, and our analysis indicated that underestimating the dimension is much more harmful than overestimating the dimension.
For this reason, the 3rd elbow was employed in the experiments performed for this work. While \citeapp{zhu2006automatic} only defines the 1st elbow, we define the $s$-th elbow as in Algorithm~\ref{algo:ZG}.

\begin{algorithm}[H]
\caption{Algorithm to compute the Zhu and Ghodsi's elbow}
\label{algo:ZG}
\begin{algorithmic}[1]
\REQUIRE The number of Zhu and Ghodsi's elbow $s$, with eigenvalues $\lambda_1, \dots, \lambda_N$
\ENSURE The $s$-th Zhu and Ghodsi's elbow
\STATE Calculate the 1st elbow $d_1$ based on $\lambda_1, \dots, \lambda_N$ according to \citeapp{zhu2006automatic}
\FOR{$i$ = 2 to s}
	\STATE {Calculate the $i$-th elbow $d_i$ based on $\lambda_{d_{i - 1} + 1}, \dots, \lambda_N$ according to \citeapp{zhu2006automatic}}
\ENDFOR
\end{algorithmic}
\end{algorithm}

Universal Singular Value Thresholding (USVT) is a simple estimation procedure proposed in \citeapp{chatterjee2015matrix} that can work for any matrix that has ``a little bit of structure''.
In the current setting, it selects the dimension $d$ as the number of singular values that are greater than a constant $c$ times $\sqrt{N/M}$.
The specific constant $c$ must be selected carefully based on the mean and variance of the entries, and since  overestimating the dimension was not overly harmful, we chose a relatively small value of $c=0.7$.

Overall, selecting the appropriate dimension is a challenging task and numerous methods could be applied successfully depending on the setting.
On the other hand, in our setting, many dimensions will yield nearly optimal mean squared errors and the two methods did not pick drastically different dimensions.
Thus efforts to ensure the selected dimension is in the appropriate range are more important than finding the best dimension.

\subsubsection{Exploration of Dimension Selection Procedures}\label{app:dim}

To further investigate the impact of the dimension selection procedures, we also considered all possible dimensions for $\hat{P}$ by ranging $d$ from 1 to $N$.
$\hat{\mathrm{MSE}}$ of $\bar{A}$ and $\hat{P}$ was plotted in Fig.~\ref{fig:realdata}.
The horizontal axis gives dimension $d$, which only impacts $\hat{P}$, which is why estimated MSE of $\bar{A}$ is shown as flat.
When $d$ is small, $\hat{P}$ underestimates the dimension and throws away important information, which leads to relatively poor performance. When $d=N$, $\hat{P}$ is equal to $\bar{A}$, so that the curve for $\hat{\mathrm{MSE}}$ for $\hat{P}$ ends at $\hat{\mathrm{MSE}}(\bar{A})$.
In the figure, a triangle denotes the 3rd elbow found by the Zhu and Ghodsi method, and a square denotes the dimension selected by USVT with threshold 0.7.
Both dimension selection algorithms tend to select dimensions which nearly minimize the mean squared error.

\begin{figure}[!htbp]
\centering
\includegraphics[width=.99\linewidth]{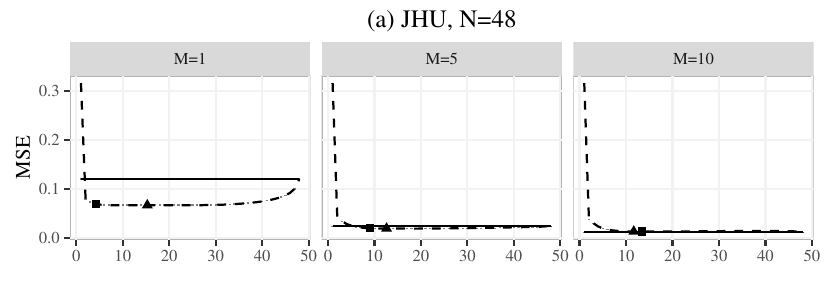}\\
\includegraphics[width=.99\linewidth]{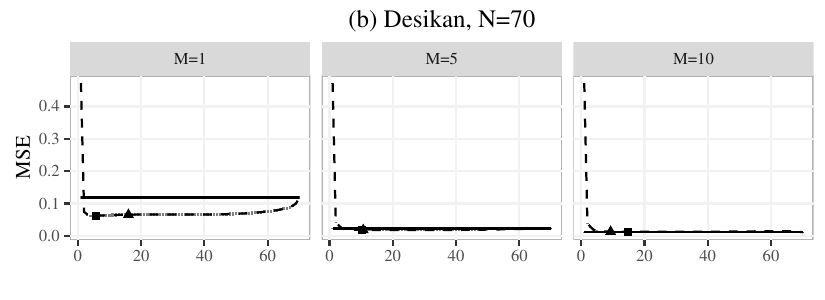}\\
\includegraphics[width=.99\linewidth]{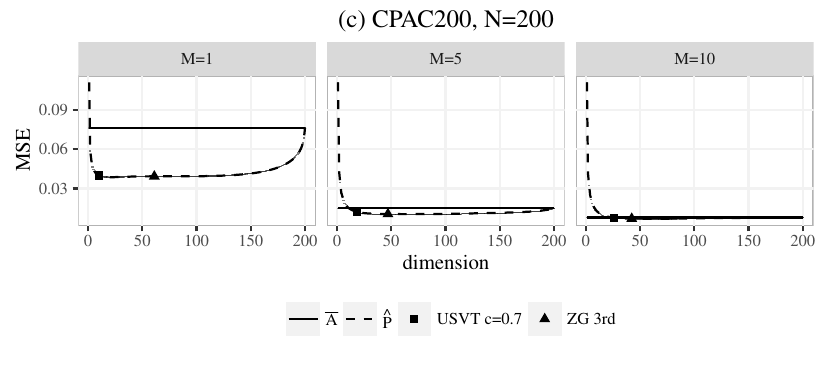}
\caption{
These plots show the mean squared error for $\bar{A}$ (solid line) and $\hat{P}$ (dashed line) for three datasets (JHU, Desikan, and CPAC200) while embedding the graphs into different dimensions and with different sample sizes $M$. The average dimensions chosen by the 3rd elbow of Zhu and Ghodsi is denoted by a triangle
 and those chosen by USVT with threshold equaling 0.7 is denoted by a square.
 Vertical intervals, visible mainly in the $N=48,70$ and $M=1$ plots, represent the 95\% confidence interval for the mean squared errors.  When $M$ is small, $\hat{P}$ outperforms $\bar{A}$ with a flexible range of the embedding dimension including the average of the dimensions selected by Zhu and Ghodsi and USVT.}
\label{fig:realdata}
\end{figure}

When $M$ is 1 or 5, $\bar{A}$ has large variance which leads to large $\hat{\mathrm{MSE}}$. Meanwhile, $\hat{P}$ reduces the variance by taking advantages of inherent low-rank structure of the mean graph. Such smoothing effect is especially obvious while there is only 1 observation. When $M = 1$, all weights of the graph are either 0 or 1, leading to a very bumpy estimate $\bar{A}$. In this case, $\hat{P}$ smooths the connectomes estimate and improves the performance.
Additionally, there is a large range of dimensions where the performance for $\hat{P}$ is superior to $\bar{A}$.
With a larger $M$, the performance of $\bar{A}$ improves so that its performance is frequently superior but nearly identical to $\hat{P}$.

\subsection{Graph Diagonal Augmentation}
\label{app:diag_aug}
The graphs examined in this work have no self-loops and thus the diagonal entries of the adjacency matrix and the mean graph are all zero.
However, when computing the low-rank approximation, these structural zeros lead to increased errors in the estimation of the mean graph.
While this problem has been investigated in the single graph setting, with multiple graphs, the problem is exacerbated since the variance of the other entries is lower, so the relative impact of the bias in the diagonal entries is higher.
Moreover, the sum of eigenvalues of the hollow matrix will be zero, leading to an indefinite matrix, which violates the positive semi-definite assumption. So it is important to remedy the situation that we do not observe the diagonal entries.

\citeapp{marchette2011vertex}  proposed the simple method of imputing the diagonals to be equal to the average of the non-diagonal entries for the corresponding row, or in equivalently the degree of the vertex divided by $n-1$.
Earlier, \citeapp{scheinerman2010modeling} proposed using an iterative method to impute the diagonal entries.
In this work, these two ideas are combined by first using the row-average method  (see Step 3 of Algorithm~\ref{algo:basic}) and then using one step of the iterative method (see Step 6 of Algorithm~\ref{algo:basic}).
Note that when computing errors, the diagonal entries are omitted since these are known to be zero.

\subsection{Evaluation of Method Choices}
\label{app:compare_param}

Each step of our algorithm is designed to improve overall performance, however for all situations these choices are not guaranteed to help.
The diagonal augmentation has negligible impact in the large-$N$ theoretical regime we've studied, but may have bigger for moderate sized $N$.
Additionally, other choices such as keeping only positive eigenvalues and the dimension selection procedure can have substantial impacts in many cases.

In this section, the claim that the choices \deleted{are} made are reasonable defaults is justified.
In particular, for the SWU4 data, we evaluated whether the particular algorithmic choices actually improve performance.
\deleted{In particular, we modified }Algorithm~\ref{algo:basic} \added{was modified} in one of four ways: 
\begin{enumerate}
\item performing the first diagonal augmentation \cite{marchette2011vertex} in step 2 or not, 
\item performing the second diagonal augmentation \cite{scheinerman2010modeling} in steps 5-6 or not, 
\item using either the Zhu and Ghodsi (ZG) procedure \cite{zhu2006automatic} or the USVT procedure \cite{chatterjee2015matrix}for dimension selection, and 
\item keeping the largest positive eigenvalues or keeping the largest eigenvalues in magnitude.
\end{enumerate}

Figure~\ref{fig:compare_param} shows the impact of each of these four binary decisions on estimating the mean graph for the SWU4 dataset for the Desikan atlas in~\ref{fig:compare_param_desikan} and for the CPAC200 atlas in~\ref{fig:compare_param_cpac200}.
\added{The two default choices are more opaque while the variations are more transparent.}
The figure demonstrates that overall, the choices that we have made \deleted{in} improve performance in these settings.

Each panel compares the impact of using either of the diagonal augmentation steps or not.
Overall, we see that the first diagonal augmentation, as indicated by color, has a small but essentially universally positive impact on relative efficiency.
It has biggest impact when the second diagonal augmentation is not performed and when $M$ is large.
The second diagonal augmentation, indicated by shape and line type, can have a a bigger impact, especially when $M$ is large.
For the CPAC200 atlas, it can happen that the second diagonal augmentation can negatively impact performance when $M$ is small.
In general however, both diagonal augmentations appear to improve performance in most cases.

The reason that diagonal augmentation has a larger impact when $M$ is large boils down to a bias-variance trade-off. 
When $M$ is large the overall variance of the entries in $\bar{A}$ is very low which means the relative impact of the bias from using zeros along the diagonal is more substantial.
Both diagonal augmentation procedures serve to reduce this bias and hence can have substantial impact when $M$ is larger.

As suggested, the impact of keeping either the largest magnitude or largest positive eigenvalues is substantial.
For most cases, our choice to keep the largest positive eigenvalues improved performance, (with the relative efficiencies generally lower for the corresponding panels).
However, when $M$ is large and USVT is used then keeping largest magnitude eigenvalues is preferable.

For the dimension selection procedure, USVT and ZG have performances which do differ for certain regimes and choices.
For the Desikan atlas, USVT appears to perform better for larger $M$ when using largest in magnitude eigenvalues.
This is likely due to the fact that USVT is using a larger dimension when $M$ is larger, which has a bigger impact when large negative eigenvalues can be included.

\begin{figure}
\begin{center}
\begin{subfigure}{.7\textwidth}
\includegraphics[width=\linewidth]{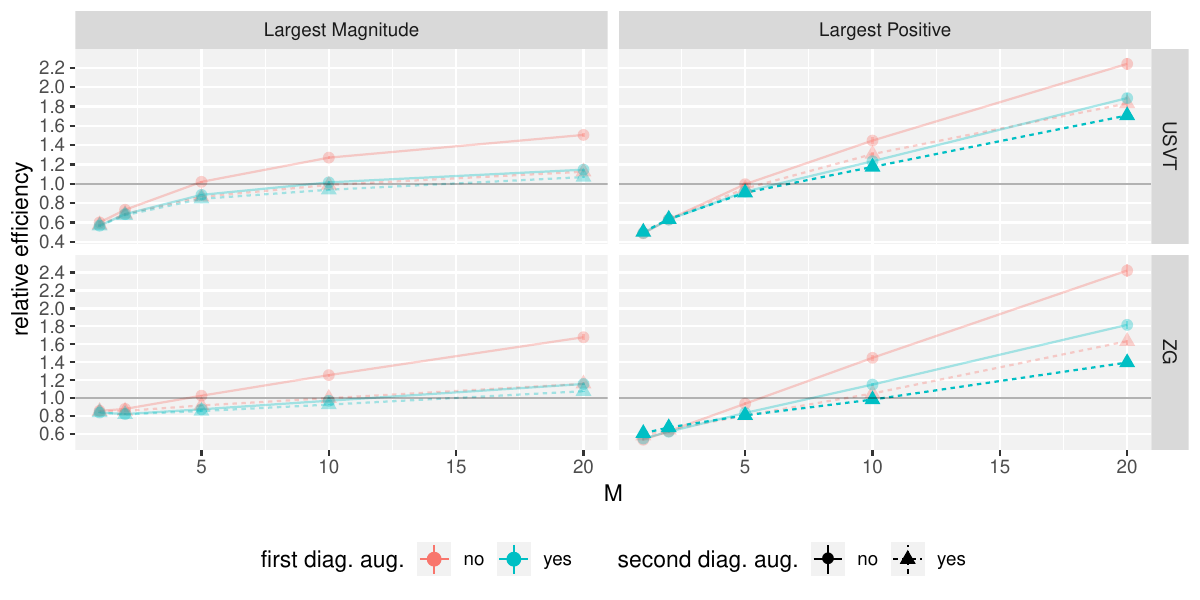}
\caption{Desikan atlas}
\label{fig:compare_param_desikan}
\end{subfigure}
\begin{subfigure}{.7\textwidth}
\includegraphics[width=\linewidth]{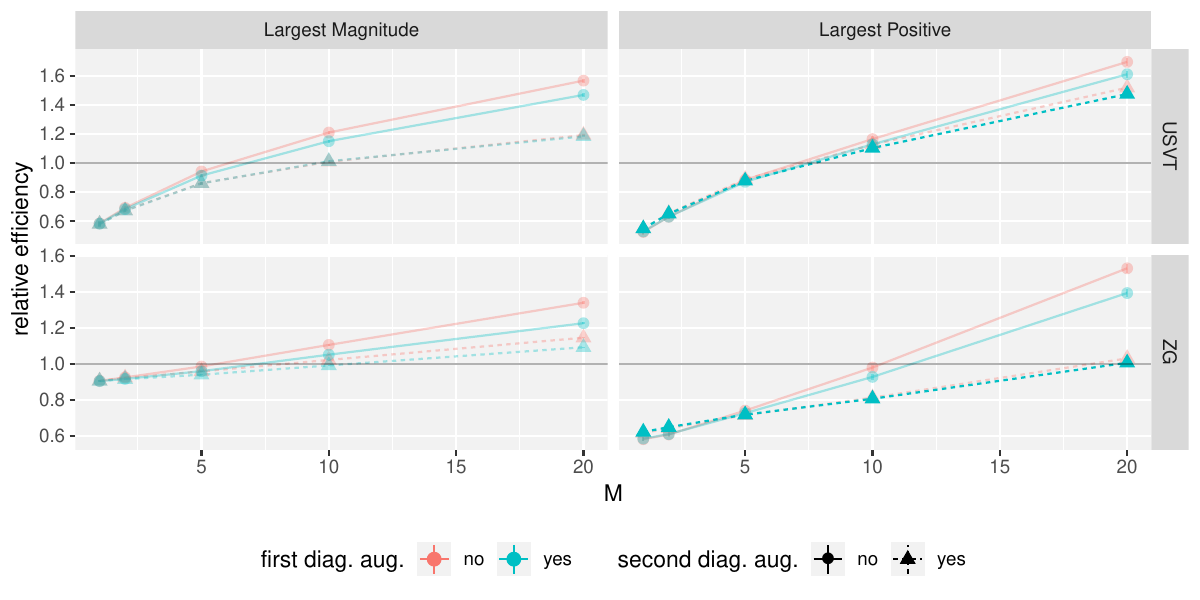}
\caption{CPAC200 atlas}
\label{fig:compare_param_cpac200}
\end{subfigure}
\end{center}
\caption{Each line depicts the relative efficiency between a version of the proposed low-rank procedure and the standard sample mean $\bar{A}$.
The top panels are for estimating the mean graph of (a) the Desikan atlas and the bottom panels for (b) the CPAC 200 atlas.
\added{The two fully opaque lines  represent the default parameters and transparent lines represent alternatives.}
The different panels indicate whether the largest eigenvalues in magnitude were kept, on the left, or the largest positive eigenvalues are kept, on the right, and whether the USVT or the ZG rank-selection procedures was used, on the top and bottom, respectively. 
The line color depicts whether the first diagonal augmentation step was used, and the linetype and shape, indicate whether the second diagonal augmentation step was used.
Note that our default setting correspond to using both diagonal augmentations, the blue dashed line with triangles, and using positive eigenvalues.
While these are not alway the best parameters we can see that they are nearly always competitive and often perform substantially better than other choices.
Further details of these experiments are provided in Appendix~\ref{app:compare_param}.}
\label{fig:compare_param}
\end{figure}

\subsection{Flipping Procedure for Permutation Test}
\label{app:testing}
Here the details of the flipping procedure are described for the permutation test mentioned in Section~\ref{section:lobe_structure}.
As mentioned before, there are 10 lobes and 70 regions based on the Desikan atlas.
We say two regions are adjacent if they share a common boundary. Such spatial adjacency is denoted by an adjacency matrix $S$ for the 70 regions, where $S_{ij} = 1$ means region $i$ and region $j$ contain a pair of voxels, $v_i$ and $v_j$, which are spatially adjacent.
If this is true, then region $j$ is defined as a neighbor of region $i$.
The lobe i.d. for region $i$ is denoted by $l_i$.

Now a uniform $1$-flip can  be defined by:
\begin{enumerate}
\item Selecting a pair of adjacent regions (region $i_1$ and region $j_1$) across the boundary of lobes uniformly, i.e. $S_{i_1 j_1} = 1$ and $l(i_1) \ne l(j_1)$;
\item Uniformly selecting another pair of adjacent regions (region $i_2$ and region $j_2$ where $i_1 \ne i_2$ and $j_1 \ne j_2$) across the same boundary of lobes uniformly, i.e. $S_{i_2 j_2} = 1$ and $l(i_1) = l(i_2)$ and $l(j_1) = l(j_2)$;
\item Reassigning region $j_1$ to lobe $l_{i_1}$ and reassign region $i_2$ to lobe $l_{j_2}$.
\end{enumerate}

By this definition, after a uniform $1$-flip, the number of regions in each lobe stays the same, where only two regions are changed to a different lobe.

Then we can define a uniform $k$-flip naturally as sequentially performing uniform $1$-flip $k$ times.
Note that after a uniform $k$-flip, the number of regions in each lobe still stays the same.

In the permutation test, a uniform $k$-flip was applied and the test statistic $T(X, l)$ was calculated based on the lobe assignment after flipping.
The $p$-value is computed as the proportion of uniform $k$-flips with a $T$ value smaller than the $T$ value for the true lobe assignments.

Figure~\ref{fig:mouse_violin} shows the violin plots for the metric defined in Eq.~\eqref{eq:test_stat} permuted in the same way as described above but applied to the mouse connectome.


\section{Dataset Description}
\label{app:data}

\subsection{Human Connectomes}
\label{app:data_human}

The original dataset is from the Emotion and Creativity One Year Retest Dataset provided by Qiu, Zhang and Wei from Southwest University available at the Consortium for Reliability and Reproducibility \citeapp{zuo2014open,gorgolewski2015high}. It is composed of 235 subjects, all of whom were college students. Each subject underwent two sessions of anatomical, resting state DTI scans, spaced one year apart. Due to incomplete data, only 454 scans are available.

When deriving MR connectomes, the NeuroData team parcellates the brain into groups of voxels as defined by anatomical atlases \citeapp{kiar2016ndmg}. The atlases are defined either physiologically by neuroanatomists (Desikan and JHU), or are generated using an automated segmentation algorithm (CPAC200).
Once the voxels in the original image space are grouped into regions, an edge is placed between two regions when there is at least one white-matter tract, derived using a tractography algorithm, connecting the corresponding two parts of the brain \citeapp{Garyfallidis2014-wg}.
The resulting graphs are undirected, unweighted, and have no self-loops.

For the Desikan atlas, there are 70 different regions (35 regions for each hemisphere), with each region belonging to a single lobe. 
Three regions of the Desikan atlas per hemisphere (Banks of Superior Temporal Sulcus, Corpus Callosum, and the ``Unknown'' region) 
do not have obvious lobe assignment and were clustered into a new lobe category named ``other'' to resolve this issue.

\subsection{Mouse Connectome}
\label{app:data_mouse}

Images of the fixed specimen were acquired on a 9.4-T small-animal-magnet using a 3D diffusion-weighted imaging sequence. 
120 unique diffusion directions were acquired using a b value of 4000 s/mm2, interleaved with 11 non-diffusion weighted scans. Images were acquired in 235 hours, and reconstructed at 43 micron resolution. The mouse brain was labeled with 296 regions of interest, 148 per hemisphere \citeapp{Anderson2017-ra}. 

To construct a structural connectome of the mouse brain fiber data was reconstructed (max 4 fiber orientations/voxel), then probabilistic tractography was performed using FSL \citeapp{behrens2007manual},
(5000 samples per voxel, 21 $\mu$m step size, 45 degrees curvature threshold). The 296 seed regions had connectivity estimates produced by counting the number of fibers that originate from one region and fall onto all other regions. This was normalized by the volume of the seed region and resulted in a 296x296 weighted, directed graph.

\section{Proofs for Theory Results}

For the proofs below we will \deleted{abuse notation and} denote $\mathrm{lowrank}_d(\replaced{\bar{A}}{P})$ as $\replaced{\tilde{P}}{\hat{P}}$.

\subsection{Outline for Main Theorems}
\label{app:outline_proof}
Here the proof of Lemma~\ref{lm:VarPhat} is outlined, which provides the approximate MSE of $\replaced{\tilde{P}}{\hat{P}}$ in the stochastic blockmodel case.
The result depends on using the asymptotic results (see Theorem \ref{thm:clt_ext}) for the distribution of eigenvectors from \citeapp{athreya2016limit} which extend to the multiple graph setting in a straightforward way.

The first key observation is that since $\bar{A}$ is computed from iid observations each with expectation $P$, $\bar{A}$ is unbiased for $P$ and $\mathrm{Var}(A_{ij}) = \frac{1}{M}P_{ij}(1-P_{ij})$.
The results of \citeapp{athreya2016limit} provide a central limit theorem for estimates of the latent position in an RDPG model for a single graph. Theorem~\ref{thm:clt_ext} describes important details.
Since the variance of each entry is scaled by $1/M$ in $\bar{A}$, the analogous result for $\bar{A}$ is that the estimated latent positions will follow an approximately normal distribution with variance scaled by $1/M$ compared to the variance for a single graph.



Since $\replaced{\tilde{P}}{\hat{P}}_{ij} = \hat{X}_i^{\top} \hat{X}_j^{\phantom{\top}}$ is a noisy version of the dot product of $\nu_s^{\top} \nu_t^{\phantom{\top}}$ from Section~\ref{section:sbm_rdpg} and each $\hat{X}_i$ is approximately independent and normal, we can use common results for the variance of the inner product of two independent multivariate normals \citeapp{brown1977means}.
After simplifications that occur in the stochastic blockmodel setting, we can derive that the variance of $\replaced{\tilde{P}}{\hat{P}}_{ij}$ converges to $\left( 1/\rho_{\tau_i} + 1/\rho_{\tau_j} \right) P_{ij} (1-P_{ij})/(N \cdot M)$ as $N \rightarrow \infty$.
Since the variance of $\bar{A}_{ij}$ is $P_{ij} (1-P_{ij})/M$, the relative efficiency between $\replaced{\tilde{P}}{\hat{P}}_{ij}$ and $\bar{A}_{ij}$ is approximately $(\rho_{\tau_i}^{-1} + \rho_{\tau_j}^{-1})/N$ when $N$ is sufficiently large.

\subsection{Proof Details}
Here the proofs are presented of the results in Section~\ref{section:theoretical_result}. 
To keep the ideas clear and concise, some details are omitted, which are only slight changes to previous works.
We assume the block memberships $\tau_i$ are drawn iid from a categorical distribution with block membership probabilities given by $\rho\in[0,1]^K$ where $\sum_i \rho_i =1$.
We will also assume that for a given $N$, the block memberships are fixed for all graphs.

We denote matrix of between-block edge probabilities by $B = \nu \nu^{\top} \in[0,1]^{K\times K}$ which we assume has rank $K$ and is positive definite.
By definition, the mean of the collection of graphs generated from this SBM is $P$, where $P_{ij} = B_{\tau_i, \tau_j}$.

We observe $M$ graphs on $N$ vertices $A^{(1)}, \cdots, A^{(M)}$ sampled independently from the SBM conditioned on $\tau$.
Define $\bar{A} = \frac{1}{M} \sum_{t=1}^M A^{(t)}$. Let $\hat{U} \hat{S} \hat{U}^{\top}$ be the best rank-$d$ positive semidefinite approximation of $\bar{A}$, then we define $\replaced{\tilde{P}}{\hat{P}} = \hat{X} \hat{X}^{\top}$, where $\hat{X} = \hat{U} \hat{S}^{1/2}$.

The proofs presented here will rely on a central limit theorem developed in \citeapp{athreya2016limit}.
The theorem was modified slightly to account for the multiple graph setting and is presented in the special case of the stochastic blockmodel.

\begin{theorem}[Corollary of Theorem 1 in \citeapp{athreya2016limit}]\label{thm:clt_ext}
  In the setting above, let $X=[X_1,\dotsc,X_N]^{\top}\in\Re^{N\times d}$ have row $i$ equal to $X_i=\nu_{\tau_i}$ (recall that $\tau_i$ are drawn from $[K]$ according to the probabilities $\rho$).
	Then there exists an orthogonal matrix $W$ such that for each row $i$ and $j$ and any $z \in \Re^{d}$, conditioned on $\tau_i=s$ and $\tau_j=t$,
  \begin{equation}
    \label{eq:4}
    \begin{split}
    &\Pr\left\{\sqrt{N}( W \hat{X}_i - \nu_s ) \leq z, \sqrt{N}( W \hat{X}_j - \nu_t) \leq z'\right\}\\
    =&  \Phi(z, \Sigma(\nu_s)/M)  \Phi(z', \Sigma(\nu_t)/M) +o(1)\end{split}
  \end{equation}
  where $\Sigma(x) =\Delta^{-1}\Ex[ X_j X_j^\top(x^\top X_j -(x^\top
  X_j)^2)]\Delta^{-1}$ and $\Delta=\Ex[ X_1 X_{1}^{T}]$ is the second
  moment matrix, with all expectations taken unconditionally.
  The function $\Phi$ is the cumulative distribution function for a multivariate normal with mean zero and the specified covariance, and $o(1)$ denotes a function that tends to zero as $N\to \infty$.
\end{theorem}
The proof of this result follows very closely the proof of the result in the original paper with only slight modifications for the multiple graph setting.

We now prove a technical lemma which yields the simplified form for the variance under the stochastic blockmodel.

\begin{lemma}
\label{lm:mseForm}
In the same setting as Theorem~\ref{thm:ARE}, for any $1 \le s, t \le K$:
\[
	\nu_s^{\top} \Sigma(\nu_t) \nu_s^{\phantom{\top}} = \frac{1}{\rho_s} \nu_s^{\top} \nu_t^{\phantom{\top}} (1- \nu_s^{\top} \nu_t).
\]
\end{lemma}
\begin{proof}
Under the stochastic blockmodel with parameters $(B, \rho)$, we have $X_i \stackrel{iid}{\sim} \sum_{k=1}^K \rho_k \delta_{\nu_k}$, where $\nu = [\nu_1, \cdots, \nu_K]^{\top}$ satisfies $B = \nu \nu^{\top}$. Without loss of generality, it can be assumed that $\nu = U S$ where $U = [u_1, \cdots, u_K]^{\top}$ is orthonormal in columns and $S$ is a diagonal matrix. Here it can be concluded that $\nu_s^{\top} = u_s^{\top} S$. Defining $R = \text{diag}(\rho_1, \cdots, \rho_K)$, allows
\[
	\Delta = \Ex[X_1 X_1^{\top}] = \sum_{k=1}^K \rho_k \nu_k \nu_k^{\top} = \nu^{\top} R \nu = S U^{\top} R U S.
\]
Thus
\begin{align*}
	\nu_s^{\top} \Sigma(\nu_t) \nu_s = &
     \sum_{k=1}^K \nu_s^{\top} \Delta^{-1} \rho_k \nu_k \nu_k^{\top} \Delta^{-1} \nu_s (\nu_t^{\top} \nu_k)(1 - \nu_t^{\top} \nu_k) \\
    = & \sum_{k=1}^K \rho_k (u_s^{\top} U^{\top} R^{-1} U u_k)^2 (\nu_t^{\top} \nu_k) (1 - \nu_t^{\top} \nu_k) \\
    = & \sum_{k=1}^K \rho_k (e_s^{\top} R^{-1} e_k)^2 (\nu_t^{\top} \nu_k) (1 - \nu_t^{\top} \nu_k) \\
    = & \sum_{k=1}^K \rho_k \delta_{sk} \rho_s^{-2} (\nu_t^{\top} \nu_k) (1 - \nu_t^{\top} \nu_k) \\
    = & \frac{1}{\rho_s} \nu_t^{\top} \nu_s (1 - \nu_t^{\top} \nu_s)
\end{align*}
\end{proof}

\begin{lemma}[Lemma~\ref{lm:VarPhat}]
In the same setting as above, for any $i, j$, conditioning on $X_i = \nu_{\tau_i}$ and $X_j = \nu_{\tau_j}$:
\[
	\lim_{N \to \infty} N \cdot \mathrm{Var}(\replaced{\tilde{P}}{\hat{P}}_{ij}) =
    \frac{1/\rho_{\tau_i} + 1/\rho_{\tau_j}}{M} P_{ij} (1 - P_{ij}).
\]
And for $N$ large enough, conditioning on $X_i = \nu_{\tau_i}$ and $X_j = \nu_{\tau_j}$:
\[
	\Ex[(\replaced{\tilde{P}}{\hat{P}}_{ij} - P_{ij})^2] \approx
    \frac{1/\rho_{\tau_i} + 1/\rho_{\tau_j}}{M N} P_{ij}(1-P_{ij}).
\]
\end{lemma}
\begin{proof}
Conditioned on $X_i = \nu_k$, we have by Theorem~\ref{thm:clt_ext},
\[
	\Ex[W \hat{X}_i] = \nu_k+o(1)
\]
and
\[
	N \cdot \mathrm{Cov}(W \hat{X}_i, W_n \hat{X}_i) = \Sigma(\nu_k)/M.
\]


Also, Corollary 3 in \citeapp{athreya2016limit} says $\hat{X}_i$ and $\hat{X}_j$ are asymptotically independent. Thus, conditioning on $X_i = \nu_s$ and $X_j = \nu_t$, we have $\lim_{N\to\infty}\Ex[\hat{X}_i^{\top} \hat{X}_j] = \lim_{N\to\infty}\Ex[(W_N \hat{X}_i)^{\top}] \Ex[W_N \hat{X}_j] = \nu_s^{\top} \nu_t = P_{ij}$.

Since $\replaced{\tilde{P}}{\hat{P}}_{ij} = \hat{X}_i^{\top} \hat{X}_j$ is a noisy version of the dot product of $\nu_s^{\top} \nu_t$, combined with Lemma~\ref{lm:mseForm} and the results above, by Equation 5 in \citeapp{brown1977means}, conditioning on $X_i = \nu_s$ and $X_j = \nu_t$:
\begin{align*}
    \Ex[\hat{X}_i^{\top} \hat{X}_j] &= \Ex[(W_N \hat{X}_i)^{\top}] \Ex[W_N \hat{X}_j] \\
    &= \nu_s^{\top} \nu_t+o(1) = P_{ij}+o(1)
\end{align*}
and
\begin{align*}
	& N \cdot \mathrm{Var} (\replaced{\tilde{P}}{\hat{P}}_{ij}) \\
    = & \frac{1}{M} \left( \nu_s^{\top} \Sigma(\nu_t) \nu_s + \nu_t^{\top} \Sigma(\nu_s) \nu_t^{\top} \right)\\
    & + \frac{1}{M^2 N} \left( tr(\Sigma(\nu_s) \Sigma(\nu_t)) \right) +o(1)\\
    = & \frac{1}{M} \left( \nu_s^{\top} \Sigma(\nu_t) \nu_s + \nu_t^{\top} \Sigma(\nu_s) \nu_t^{\top} \right)+o(1) \\
    = & \frac{1/\rho_s + 1/\rho_t}{M} P_{ij}(1-P_{ij}) + o(1).
\end{align*}
Since $\replaced{\tilde{P}}{\hat{P}}_{ij} = \hat{X}_i^{\top} \hat{X}_j$ is asymptotically unbiased for $P_{ij}$, when $n$ is large enough:
\[
    \Ex[(\replaced{\tilde{P}}{\hat{P}}_{ij} - P_{ij})^2] = \mathrm{Var}(\replaced{\tilde{P}}{\hat{P}}_{ij}) \approx
    \frac{1/\rho_s + 1/\rho_t}{M N} P_{ij}(1-P_{ij})+o(1).
\]
\end{proof}

We now prove Theorem~\ref{thm:ARE}
\begin{proof}[Proof of Theorem~\ref{thm:ARE}]
Combining the MSE result of $\bar{A}_{ij}$
\[
    \Ex[(\bar{A}_{ij} - P_{ij})^2] = \frac{P_{ij}(1-P_{ij})}{M},
\]
and Lemma \ref{lm:VarPhat}, i.e. for large enough $N$,
\[
    \Ex[(\replaced{\tilde{P}}{\hat{P}}_{ij} - P_{ij})^2] \approx
    \frac{1/\rho_{\tau_i} + 1/\rho_{\tau_j}}{M N} P_{ij}(1-P_{ij}),
\]
and therefore there is a large enough $N$,
\[
	    \mathrm{RE}(\bar{A}_{ij}, \replaced{\tilde{P}}{\hat{P}}_{ij}) 
	    = \frac{\Ex[(\replaced{\tilde{P}}{\hat{P}}_{ij} - P_{ij})^2]}{\Ex[(\bar{A}_{ij} - P_{ij})^2]}
	    \approx \frac{1/\rho_{\tau_i} + 1/\rho_{\tau_j}}{N}.
\]
And the ARE result follows directly by taking the limit of RE as $N\to \infty$.
\end{proof}


The proof for Theorem~\ref{thm:ARE} is now a simple application of the above lemmas to the ratio of the mean squared errors for $\bar{A}$ and $\hat{P}$.

\section{SBM Simulations} \label{app:sbm_sim}

\begin{figure}[!htbp]
    \centering
    \includegraphics[width=1\linewidth]{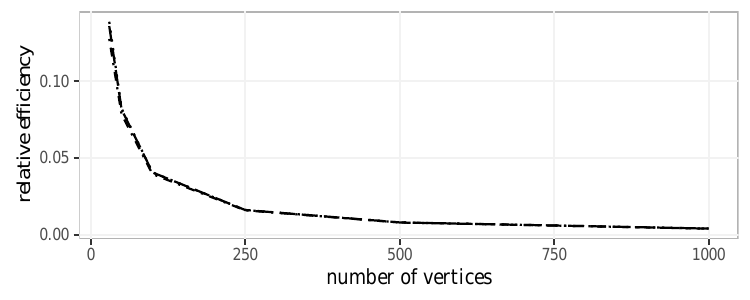}
    \includegraphics[width=1\linewidth]{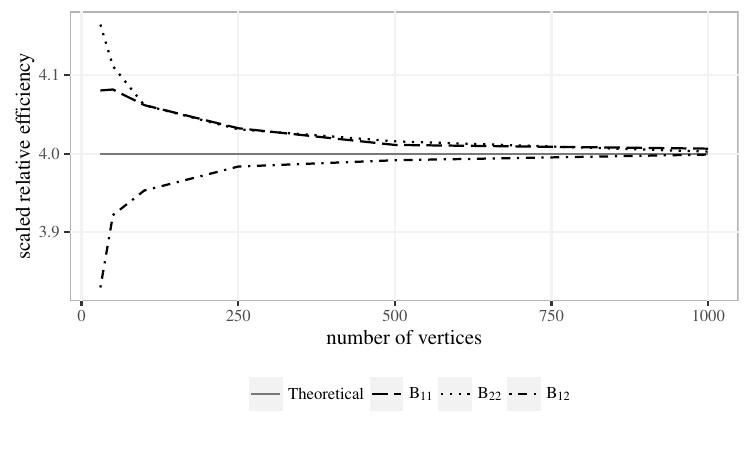}
    \caption{
    The top panel shows the estimated relative efficiency $\hat{\mathrm{RE}}(\bar{A},\hat{P})$ as a function of $N$ for fixed $M=100$ based on simulations of an SBM.
    For each value of $N$, 1000 Monte Carlo replicates of the SBM from Appendix~\ref{app:sbm_sim} estimated the RE.
    Each curve corresponds to an average across vertex pairs corresponding to the three distinct block probabilities $B_{11}$, $B_{12}$, and $B_{22}$ in the two-block SBM.
    Recall that values below 1 indicate that $\hat{P}$ is performing better than $\bar{A}$.
    To distinguish the three curves, the bottom panel shows the corresponding scaled relative efficiencies, $N\cdot \hat{\mathrm{RE}}(\bar{A},\hat{P})$.
    The solid horizontal line indicates the theoretical asymptotic scaled relative.}
    \label{fig:RE}
\end{figure}

In this section, the theoretical results from Section~\ref{section:theoretical_result} regarding the relative efficiency between $\bar{A}$ and $\hat{P}$ \textit{via} Monte Carlo simulation experiments in an idealized setting will be illustrated.
These numerical simulations will also allow us to investigate the finite sample performance of the two estimators.


Here, we consider the following 2-block SBM with parameters
\begin{equation*}
B = \begin{bmatrix}
0.42 & 0.2 \\
0.2 & 0.7
\end{bmatrix}
,\qquad \rho = \begin{bmatrix}
0.5 & 0.5
\end{bmatrix}.
\end{equation*}
When calculating $\hat{P}$, the dimension selection step from Algorithm~\ref{algo:basic} is omitted and replaced with the true dimension $d = \mathrm{rank}(B) = 2$.
Note that for large $N$, many dimension selection methods will often correctly select the true dimension \cite{chatterjee2015matrix,Fishkind2012}.
1000 Monte Carlo replicates were performed with the above SBM distribution with $N \in \{30, 50, 100, 250, 500, 1000 \}$

Since the relative efficiency only depends on the block memberships of the pair $i,j$, letting $D_{st} = \{(i, j): \tau_i=s,\tau_j=t,1 \le i < j \le n\}$, the relative efficiency for each block pair can be estimated using
\[
    \hat{\mathrm{RE}}_{st}(\bar{A},\hat{P}) = \frac{\sum_{(i, j) \in D_{st}} \hat{\mathrm{MSE}}(\hat{P}_{ij})}{\sum_{(i, j) \in D_{st}} \hat{\mathrm{MSE}}(\bar{A}_{ij})}
\]
for $s,t\in\{1,2\}$, where $\hat{\mathrm{MSE}}$ denotes the estimated mean squared error based on the Monte Carlo replicates.
For the remaining simulations and real data analysis, we will always be considering estimated relative efficiency and estimated mean squared error rather than analytic results, and hence we will frequently omit that these are estimated values. 

In Fig.~\ref{fig:RE}, we plot the (estimated) relative efficiency (top panel) and the scaled (estimated) relative efficiency (bottom panel), $N \cdot \hat{\mathrm{RE}}_{st}(\bar{A},\hat{P})$.
The different dashed lines denote the RE and scaled RE associated with different block pairs, either $B_{11}$, $B_{12}$, or $B_{22}$.
As expected from Theorem~\ref{thm:ARE}, the top panel indicates that the relative efficiencies are all very close together and much less than 1, decreasing at the rate of $1/N$, indicating that $\hat{P}$ is performing better than $\bar{A}$.

Based on Theorem~\ref{thm:ARE}, the scaled RE converges to $1/\rho_{\tau_i}+1/\rho_{\tau_j}=4$ as $N\to\infty$ for all pairs $i,j$.
This is plotted as a solid line in the bottom panel.
The figure shows that $N \cdot \hat{\mathrm{RE}}_{st}(\bar{A}, \hat{P})$ converges to scaled asymptotic RE quite rapidly.
Error bars were omitted, as the standard errors are very small for these estimates.

\begin{remark}
For small graphs, the estimates of the edge probabilities for pairs of vertices in different blocks are much better than the estimates for edges within each block.
The reason for this is unclear and could be due to the actual values of the true probability, but it may also be due to the fact that there are approximately twice as many pairs of vertices in different blocks, $N^2/4$, than there are in the same block, $N^2/8-N/4$.
This could lead to an increase in effective sample size which may cause the larger differences displayed in the left parts of Fig.~\ref{fig:RE}.
However, these differences are nearly indistinguishable for unscaled relative efficiency overall.
\end{remark}


\section{Synthetic Data Analysis for Full Rank IEM}\label{app:sim_iem}

While the theory we have developed is based on the assumption that the mean graph is low rank, as we have seen in Section~\ref{sec:connectome}, $\hat{P}$ can perform well even when this assumption is false.
To further illuminate this point, a synthetic data analysis under a more realistic full-rank independent edge model was performed.
As discussed in Section~\ref{sec:challenge}, the sample mean of the 454 graphs in the Desikan dataset is actually of full rank.
For this simulation, we will use the sample mean as the probability matrix $P$.
A sampling of independent graphs from the full rank IEM with the probability matrix $P$ show that for the synthetic data sets of size $M = 1, 5$, and $10$, $\hat{P}$ performs even better than $\bar{A}$ in the real data experiments. 
Fig.~\ref{fig:sim_desikan} shows the resulting estimated MSE for $\bar{A}$ (solid line) and $\hat{P}$ (dashed line), as a function of the embedding dimension for simulated data based on the full rank probability matrix $P$ shown in the left panel of Fig.~\ref{fig:Matrix_desikan_m5}.
These results are similar to those presented in Section~\ref{sec:connectome}, though overall $\hat{P}$ performs even better than in the real data experiments.
When $M$ is small, $\hat{P}$ outperforms $\bar{A}$ with a flexible range of embedding dimensions including those selected by the Zhu and Ghodsi method.
On the other hand, when $M$ is large enough, both estimators perform well with the decision between the two being less conclusive.
This simulation again shows the robustness of $\hat{P}$ to deviations from the RDPG model, specifically if the probability matrix is full-rank.

We also note that the finite-sample relative efficiency for this synthetic data is even more favorable to $\hat{P}$ than for the real data, with relative efficiency lower than $1/3$ for $M=1$ in the synthetic data analysis as compared to relative efficiency which were at best around $1/2$ for $M=1$ in the original data.
From this observation, we can postulate that the degradation in the performance of $\hat{P}$ in real data can at least partially be attributed to the fact that the independent edge assumption does not hold for real data.
It also suggests that more elaborate models of connectomes will be valuable for various inferential tasks.
{

\begin{figure}[!htbp]
\centering
\includegraphics[width=1\linewidth]{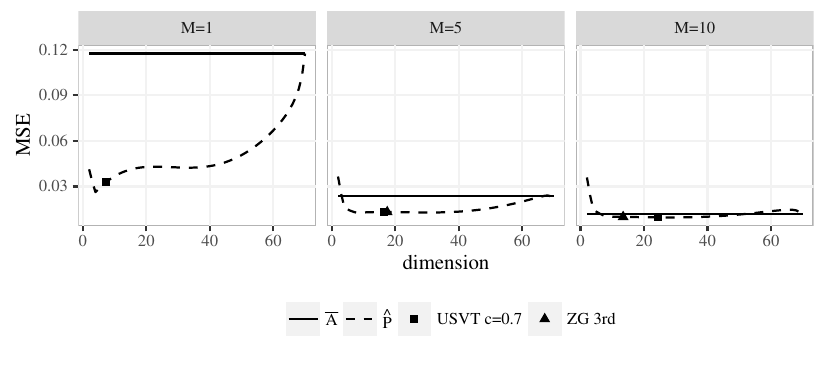}
\caption{
Comparison of $\hat{P}$ and $\bar{A}$ for synthetic data analysis.
As in Fig.~\ref{fig:RE}, this figure shows $\hat{\mathrm{MSE}}$ for $\bar{A}$ (solid line) and $\hat{P}$ (dashed line) for simulated data with different sample sizes $M$ based on the sample mean for the Desikan dataset. Again, the average of dimensions selected by the USVT method (square) and the ZG method (triangle) tend to nearly approximate the optimal dimension.
Overall, the structure of these plots well approximates the structure for the real data indicating that performance for the independent edge model will tend to translate in structure to non-independent edge scenarios.
On the other hand, the relative efficiency $\hat{\mathrm{RE}}(\bar{A},\hat{P})$ is lower for this synthetic data analysis than for the SWU4 data.}
\label{fig:sim_desikan}
\end{figure}

\section{Analysis of Mouse Superstructures}

Fig.~\ref{fig:mouse_violin} shows the analogous permutation analysis to that performed in Section~\ref{section:lobe_structure} and Fig.~\ref{fig:violin_plot} but for the mouse data described in Appendix~\ref{app:data_mouse}.
\begin{figure}[tbh!]
	\centering
	\includegraphics[width=\linewidth]{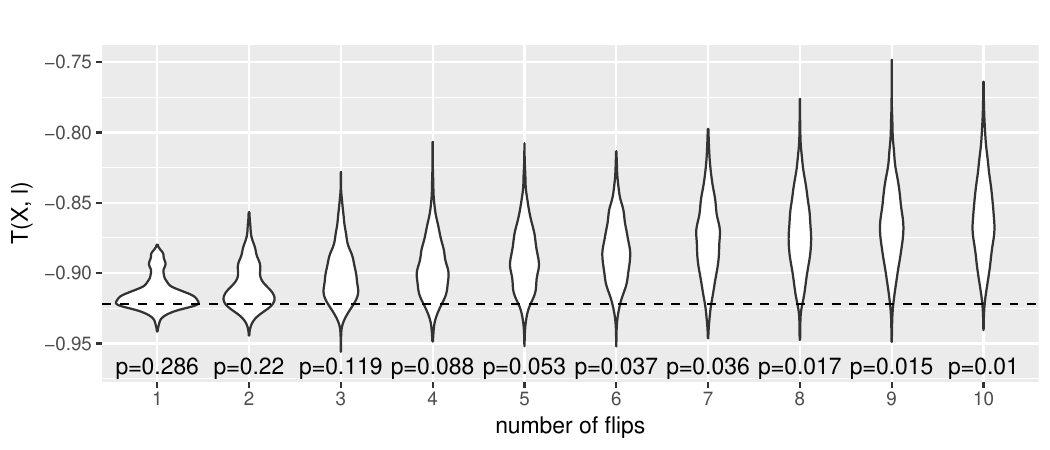}
	\caption{Violin plot for the mouse connectome permutation test. As in Fig.~\ref{fig:violin_plot}, 1000 simulations for each number of flips show that as the number of flips increases, the distribution under the null moves further from the dashed line. This indicates that the latent positions correlate with the superstructures even after conditioning on the spatial structure. Dashed line represents the situation based on true superstructure assignment.   }
	\label{fig:mouse_violin}
\end{figure}

\bibliographystyleapp{IEEEtran}

\bibliographyapp{Bib}

\end{document}